\documentclass[english]{article}
\usepackage{xcolor}
\usepackage{graphicx}
\usepackage{geometry}
\usepackage[T1]{fontenc}
\usepackage[latin9]{inputenc}
\usepackage{amsbsy}
\usepackage{caption}
\usepackage{subcaption}
\usepackage{comment}
\usepackage{setspace}
\usepackage{multirow}
\usepackage{hyperref}
\usepackage[authoryear,round]{natbib}
\usepackage{amsmath}
\usepackage{amssymb}
\usepackage{booktabs}
\usepackage{bbm}
\usepackage{algorithm}
\usepackage{algpseudocode}
\usepackage[normalem]{ulem}
\usepackage{amsthm}
\usepackage{tikz}
\usetikzlibrary{decorations.markings}
\newtheorem{prop}{Proposition}
\newtheorem{lemma}{Lemma}

\newcommand{\RNum}[1]{\uppercase\expandafter{\romannumeral #1\relax}}

\newcommand{\parentheses}[1]{\left ( #1\right )}

\begin{document}
\title{Locally adaptive step sizes for Hamiltonian Monte Carlo methods}

\author{Jimmy Huy Tran\footnote{Corresponding author, email: jimmy.tran@uis.no, Department of Mathematics and Physics, University of Stavanger, 4036 Stavanger, Norway} 
\and Tore Selland Kleppe\footnote{Department of Mathematics and Physics, University of Stavanger, 4036 Stavanger, Norway}
\and Nawaf Bou-Rabee\footnote{Department of Mathematical Sciences, Rutgers University, 311 N 5th Street,
Camden, NJ 08102, USA}}
\date{}
\maketitle


\begin{abstract}
Exact reversible locally adaptive step sizes are introduced by \cite{bou2025within} in the context of the No-U-Turn Hamiltonian Monte Carlo (HMC) sampler for the purpose of high performance HMC sampling in challenging target distributions. This paper carefully explains and illustrates the locally adaptive step size procedure under more general assumptions. Alternative error measures are introduced with the aim of further enhancing MCMC performance by reducing the prevalence of reversibility failure rejections. In addition, the paper shows how to introduce the locally adaptive step sizes also within standard HMC and Metropolis-corrected kinetic Langevin/generalized HMC (GHMC). Finally, a numerical study of the relative merits of the proposed error measures in the context of different Hamiltonian flow-based MCMC methods is provided.
\end{abstract}



\doublespacing

\section{Introduction} \label{sec:intro}

Markov chain Monte Carlo (MCMC) \citep{GelmanBDA3} is a cornerstone of modern Bayesian analysis when the target posterior density is not analytically computable. In recent years, gradient-based MCMC methods have risen in popularity, mainly due to the target-informed proposals. Several of these algorithms are based on numerically integrated Hamiltonian flows. Such flows are naturally incorporated in the MCMC framework as they have the ability to preserve a desired target distribution after an appropriate Metropolis correction. Thus, a range of applications in fields like theoretical physics, molecular simulations and Bayesian statistics have adopted these methods. Examples of prominent algorithms in this family include kinetic/second-order/underdamped Langevin algorithms \citep{Leimkuhler2015molecular, durmus2025uniform}, which include generalized Hamiltonian Monte Carlo \citep{horowitz1991generalized}, standard Hamiltonian/Hybrid Monte Carlo \citep{duane1987hybrid, mackenze1989improved, neal2011mcmc} and highly specialized Bayesian inference MCMC algorithms based on the former with the No-U-Turn sampler (NUTS) \citep{hoffman2014no} as the prominent example. 

The methods above typically require some sort of numerical integration to approximate the Hamiltonian flow. To this end, the Verlet/St{\o}rmer leapfrog integrator is typically used in conjunction with fixed (time) step sizes to guarantee time reversibility of the numerically integrated flow. However, to ensure a successful and efficient numerical integration procedure, the fixed step sizes must be tailored explicitly to each particular target distribution. Usually, the step size is chosen so that the numerically approximated flow is accurate enough even in the parts of the support where the logarithm of the target density has the largest curvature. Consequently, the parts of the support with more moderate curvature will be traversed using unnecessarily small step sizes, which in turn leads to wasted computational effort. In contrast, if the fixed step sizes are too large, this may lead to very slow exploration of the higher curvature regions, as stability limits may be breached for almost all simulated trajectories and subsequently lead to divergent transitions \citep{stan2018}.

Adaptive step sizes for the numerical integration of ordinary differential equations are ubiquitous in resolving the weaknesses mentioned in the preceding paragraph \citep[see e.g.,][]{hairer2008solving}. Nevertheless, the naive use of such step size controllers to Hamiltonian involution flows obtained by numerical integrators would generally lead to the loss of both reversibility/involutional nature and the volume preservation property. With numerical flows no longer being involutions, it is not clear how to proceed with Metropolis-correcting for incurred numerical errors (even if the Jacobian determinant of the numerical flow is available). \cite{kleppe2022connecting} considers uncorrected MCMC based on Hamiltonian dynamics integrated using such conventional step size controllers. This approach may lead to competitive performance for challenging target distributions and numerical results indicate only minor biases relative to the overall Monte Carlo variation. However, it appears very difficult to quantify the incurred biases a priori.

Building on the adaptive step size mechanism of the recently proposed WALNUTS sampler \citep{bou2025within}, this paper considers adaptive step sizes for Hamiltonian flow-based MCMC while at the same time retaining reversibility/correctness. The first contribution involves both a more detailed description and further extensions to the basic energy error-based algorithm of \cite{bou2025within} by leveraging insights from the numerical ODE literature.

More concretely, the adaptive step size algorithm of \citet{bou2025within} proceeds by dividing the desired numerically approximated Hamiltonian flow into a sequence of \emph{macro} steps with fixed (time-) step size. Each macro step is traversed by a varying (integer) number of \emph{micro} steps. The number of micro steps in each macro step is determined to ensure sufficient accuracy and therefore provides the desired adaptivity.

However, to ensure reversibility, the ideal procedure for determining the number of micro steps in each macro step should be such that the same number of micro steps would be chosen regardless of whether the integration happens forward or backward in time. Such exact reversibility is in general not feasible and/or not computationally efficient. To account for such imperfect reversibility, macro steps that fail to be reversible must either be rejected or the corresponding micro step numbers must be randomized \citep[see also][]{bou2024incorporating}. In both cases, this will impact the performance of the overall algorithm. To reduce the frequency of such reversibility breaches, the present paper describes further procedures to select the number of micro steps in each macro step beyond the energy error-based criterion of \cite{bou2024incorporating,bou2025within}.

A further contribution of the paper is a detailed description of how such adaptive step size mechanisms may also be incorporated in other Hamiltonian flow-based MCMC algorithms beyond WALNUTS. Like for WALNUTS, the recently proposed Gibbs Self-tuning (GIST) framework of \citet{bou2024gist} is utilized to this end. More specifically, it is demonstrated that these alternative adaptive step size methods can be formulated as a GIST sampler. Consequently, the corresponding transition kernel will be reversible with respect to the target distribution, thus ensuring the correctness of the methods. 

A key finding of this paper is that the benefit of the proposed alternative micro step controllers depends strongly on the overarching Hamiltonian flow-based MCMC it is used within. Potentially larger benefits are observed for Hamiltonian flow-based MCMC methods where the complete proposal must be discarded if a reversibility failure occurs.

The remainder of this paper is laid out as follows: Section \ref{sec:background} provides the background on HMC, GIST and the macro/micro locally adaptive step sizes of \cite{bou2025within}. Section \ref{sec:micro_step_criterion} introduces and compares different error measures under the framework of locally adaptive step size selection with the aim of reducing reversibility failures. Section \ref{sec:other_HMC_methods} explains how to include locally adaptive step sizes within other Hamiltonian flow-based MCMC methods beyond WALNUTS. Finally, Section \ref{sec:numerical_ill} provides numerical illustrations and Section \ref{sec:discussion} concludes with a discussion. Source code for the implementations and numerical experiments is available at \url{https://github.com/jihut/locally_adaptive_step_size_hmc}. 

\section{Background and setup} \label{sec:background}

The main goal is to sample from a given target distribution characterized by a probability density function (PDF) $\pi(\theta) := \Tilde{\pi}(\theta) / Z,\;\theta \in \mathbb{R}^d$. Here, $1 / Z$ represents the unknown normalizing constant while $\Tilde{\pi}(\theta)$ is a density kernel of the PDF that can be evaluated pointwise. 


\subsection{Standard Hamiltonian Monte Carlo} \label{sec:standard_hmc}
The idea of the standard HMC method is to extend the state space $\mathbb R^d$ to the phase space $\mathbb R^{2d}$ by introducing an auxiliary momentum variable $\rho \in \mathbb R^d$ so that $z = (\theta, \rho) \in \mathbb R^{2d}$. Furthermore, let $M \in \mathbb R^{d \times d}$  be a symmetric and positive definite mass matrix. An ideal HMC method proceeds by constructing a Markov chain using proposals obtained by evolving the Hamiltonian system using the Hamiltonian flow $\varphi_t$ given as $\varphi_t(z(\tau)) := z(\tau + t)$. This flow preserves the joint density $\mu(\theta, \rho) \propto \exp(-H(\theta, \rho))$, which is also commonly known as the density of the Boltzmann-Gibbs distribution. Here, $H$ simply denotes the Hamiltonian of the system at a given state $z$: 
\begin{equation*} \label{eq:hamiltonian_def}
    H(z)=H(\theta, \rho):= - \log \Tilde{\pi}(\theta) + \frac{1}{2} \rho^T M^{-1} \rho
\end{equation*}
The evolution of the state in time, and hence related to the flow $\varphi_t$, is obtained from solving Hamilton's equations: 
\begin{align}
    {\Dot{\theta}}(t) &= M^{-1} \rho(t) \label{eq:hamiltonian_ode_nr1}, \\
    {\Dot{\rho}}(t) &= \nabla_{\theta} \log \Tilde{\pi}(\theta(t)) \label{eq:hamiltonian_ode_nr2}
\end{align}

As Equations \eqref{eq:hamiltonian_ode_nr1} and \eqref{eq:hamiltonian_ode_nr2} are not analytically solvable in most cases, the flow must be approximated numerically using a numerical integrator. The standard choice of such an integrator is the (Verlet/St{\o}rmer) leapfrog integrator \citep[see e.g.,][]{Leimkuhler:2004,bou-rabee_sanz-serna_2018} as it preserves several key traits of the true Hamiltonian flow that are required to construct a simple but reversible Markov chain based on these proposals. More specifically, assume a current initial state in the phase space simply denoted as $z = (\theta, \rho)$. Then, given a fixed step size $h>0$, the leapfrog integrator $\Phi_h$ evolves the system by $h$ time units from $z$ to $z' = \Phi_h(z)$ in the following way: 
\begin{align}
    \theta' &= \theta + M^{-1} \parentheses{h \rho + \frac{h^2}{2} \nabla_{\theta} \log \Tilde{\pi}(\theta)} \label{eq:leapfrog_update_theta}, \\
    \rho' &= \rho + \frac{h}{2}(\nabla_{\theta} \log \Tilde{\pi}(\theta) + \nabla_{\theta} \log \Tilde{\pi}(\theta')) \label{eq:leapfrog_update_rho}
\end{align}
Some of the most important properties that the leapfrog integrator has in common with $\varphi_t$ are:
\begin{itemize}
    \item Volume preservation: The leapfrog map preserves the volume in the phase space so that the determinant of the Jacobian of the map is 1.
    \item Inverse map: The inverse map is the same as evolving back in time, i.e., $\Phi_h^{-1} = \Phi_{-h}$. 
    \item Time reversibility: Let $\mathcal F$ denote the momentum flip operator, i.e., $\mathcal F(z)=(\theta, -\rho)$. Then, the inverse map can be expressed as 
    \begin{equation} \label{eq:leapfrog_time_reversibility}
        \Phi_h^{-1} = \mathcal F \circ \Phi_h \circ \mathcal F. 
    \end{equation}
    Based on this fact, the map $\mathcal F \circ \Phi_h$ is therefore an involution, i.e., the inverse of this map is itself. 
\end{itemize}
By extending to a situation where one evolves the initial state by $L \in \mathbb Z$ leapfrog steps, it is convenient to denote this composition of maps as $\Phi_h^L$. Then, one can express the map recursively based on the definition of the single leapfrog map as
\begin{equation*}
    \Phi_h^{L+1}:= \Phi_h \circ \Phi_h^{L}, \quad \Phi_h^{L-1}:= \Phi_h^{-1} \circ \Phi_h^{L}. 
\end{equation*}
This notation makes it possible to take both $L >0$ leapfrog steps forward in time and $L<0$ steps backward, which is useful when defining the more sophisticated HMC-like samplers later in the text. Given that the step size $h$ is fixed, the map $\Phi_h^L$ will also inherit the property of time reversibility. 

Based on the properties of the leapfrog integrator, a simple Markov kernel that is $\mu$-invariant and uses proposals from the Hamiltonian system can be defined by employing the involution-based Metropolis-Hastings-Green algorithm \citep{green1995reversible} in the following way: 
\begin{enumerate}
    \item  Momentum refresh: Sample $\rho \sim N(0, M)$. 
    \item Given $\rho$ and the previous position state $\theta$, evolve this state by $L$ leapfrog steps so that $(\theta', \rho') = \Phi_h^{L} (\theta, \rho)$. 
    \item Although not necessary for a situation with full momentum refresh after every iteration like in the standard HMC setting currently considered, the Metropolis-Hastings-Green algorithm requires that the deterministic map that leads to the proposed state is an involution. Thus, based on \eqref{eq:leapfrog_time_reversibility}, the relevant proposed state is $z^* = (\theta', -\rho')$. 
    \item Using the volume preservation property, the acceptance probability given by the Metropolis-Hastings-Green algorithm simply becomes 
    \begin{equation*}
        \alpha(z, z^*) := \min \parentheses{1, \exp(H(z) - H(z^*))}.
    \end{equation*}
    Then, the new proposed state $z^*$ is accepted if $u < \alpha(z, z^*)$ and rejected otherwise, where $u \sim \text{Uniform[0,1]}$. 
\end{enumerate}
The above constitutes the standard HMC sampler and will be the foundation of the more elaborate samplers later. 

For the standard HMC setup, it is crucial that the step size $h$ is fixed for the $L$ leapfrog steps that are performed (or more generally, at least the array consisting of the corresponding step sizes at each leapfrog step should be palindromic) in order to obtain a valid proposal. Otherwise, the time reversibility of $\Phi_h^L$ will be violated and the map $\mathcal F \circ \Phi_h^L$ is not an involution so that the involution-based Metropolis-Hastings-Green algorithm is no longer valid.

In particular, suppose $h=\gamma(\theta,\rho)$ is some non-trivial step size selection mechanism, and let $\hat \Phi(\theta,\rho) = \Phi_{\gamma(\theta,\rho)}(\theta,\rho)$. For a given $(\theta,\rho)$, denote $(\theta',\rho')=\hat \Phi(\theta,\rho)$. Note that for $\mathcal F \circ \hat \Phi$ to be an involution, one would require that $\gamma(\theta,\rho)=\gamma(\theta',-\rho')$, i.e., a reversible step size controller, which to the authors' current knowledge is not readily available. To avoid these problems, one can proceed by representing the step sizes themselves as additional auxiliary random variables, which is the core idea of the GIST framework.

\subsection{HMC-based GIST basics}
\label{sec:GIST-basics}
In the context of Hamiltonian flow-based proposals, a GIST sampler (see \citet{bou2024gist} for a more detailed and rigorous description) is a (Metropolized) Gibbs sampler targeting an enlarged space consisting of $(\theta,\rho)$ and the tuning parameters used for determining the Hamiltonian flow. Specifically, denote $v \in \mathbb V$ as the collection of auxiliary/tuning parameters for the flow step (e.g., the number of leapfrog steps, step sizes, mass matrices, etc.), and suppose that a tuning parameter distribution $p(v|\theta,\rho)$ has been selected. Then, the GIST sampler proceeds to target the joint density $\mu(\theta,\rho,v)\propto \exp(-H(\theta,\rho))p(v|\theta,\rho)$ on the augmented space $\mathbb A = \mathbb R ^ {2d} \times \mathbb V$.

The final ingredient of a GIST specification is a mapping of the form
\begin{equation}
    ((\theta',\rho'),v') =(\mathcal F \circ G_v(\theta,\rho),g_{(\theta,\rho)}(v)),
\label{eq:genGISTInvolution}
\end{equation}  
where $G_v$ is a numerically integrated Hamiltonian flow with tuning parameters $v$ and $g_{(\theta,\rho)}$ is a ``parameter mapping.'' 
Throughout the text, $\mathbb V$ is a countable set (which corresponds to a set of vectors of micro step counts in this paper), and the reference measure on the augmented space $\mathbb A=\mathbb R^{2d}\times\mathbb V$ is the product measure
\begin{equation*}
  \nu := \mathrm{Leb}_{2d}\otimes \#_{\mathbb V},
  \label{eq:reference_measure}
\end{equation*}
where $\mathrm{Leb}_{2d}$ denotes the Lebesgue measure on $\mathbb R^{2d}$ and $\#_{\mathbb V}$ the counting measure on $\mathbb V$. For each $(\theta,\rho)$, $p(\cdot\mid\theta,\rho)$ is a probability mass function on $\mathbb V$, and $\mu(\theta,\rho,v)$ is a probability density with respect to $\nu$. The complete mapping in \eqref{eq:genGISTInvolution} must be a $\nu$-preserving involution of $\mathbb A$. In particular, it must satisfy
\begin{align*}
    \mathcal F
    \circ G_{g_{(\theta,\rho)}(v)}
    \circ \mathcal F
    \circ G_v(\theta,\rho)
    &=
    (\theta,\rho), \\
    g_{\mathcal F\circ G_v(\theta,\rho)}
    \left(g_{(\theta,\rho)}(v)\right)
    &=
    v.
\end{align*}
If the complete mapping in \eqref{eq:genGISTInvolution} is not
reference-measure preserving, the corresponding Radon--Nikodym derivative must instead be included in the Metropolis-Hastings acceptance probability.
A sufficient condition for $\nu$-preservation, which covers every construction in this paper, is that the parameter mapping does not depend on the state, i.e., $g_{(\theta,\rho)}=g$ for a bijection $g$ of $\mathbb V$, and that for every fixed $v$, the map $G_v$ is a volume preserving bijection of $\mathbb R^{2d}$. Indeed, writing $\Psi(\theta,\rho,v):=((\mathcal F\circ G_v)(\theta,\rho),g(v))$, one has $\Psi^{-1}(A\times\{v'\})=\{(\theta,\rho,g^{-1}(v')):(\mathcal F\circ G_{g^{-1}(v')})(\theta,\rho)\in A\}$ for every measurable $A\subset\mathbb R^{2d}$ and $v'\in\mathbb V$, and this set has $\nu$-measure $\mathrm{Leb}_{2d}(A)$. When the macro step size is randomized (like in Section~\ref{sec:other_HMC_methods}), the step size vector is included in $v$, $\mathbb V$ becomes the product of a countable set and a Euclidean space, and $\#_{\mathbb V}$ is replaced by the product of counting and Lebesgue measure while nothing else changes.

An instructive non-trivial example of such a joint mapping would be two leapfrog steps with individual step sizes, i.e.,  $v=(h_1,h_2)$, $G_v = \Phi_{h_2} \circ \Phi_{h_1}$ and $g_{(\theta,\rho)}(h_1,h_2)=(h_2,h_1)$ (so the order of the step size array is reversed). Then, it is straightforward to verify that (\ref{eq:genGISTInvolution}) is an involution as $\mathcal F \circ G_{g_{(\theta,\rho)}(v)} \circ \mathcal F \circ G_v(\theta, \rho) =  \mathcal F \circ \Phi_{h_1} \circ \Phi_{h_2} \circ \mathcal F \circ \Phi_{h_2} \circ \Phi_{h_1} (\theta, \rho) = (\theta, \rho)$ and $g_{\mathcal F\circ G_v(\theta,\rho)}\circ g_{(\theta,\rho)}=\mathrm{id}$, since $g$ does not depend on the state and reverses the order twice.

Provided that the above elements are specified, a GIST sampler would perform the following steps: 
\begin{enumerate}
    \item (Partial) momentum refresh that preserves $\rho \sim N(0,M)$. 
    \item Update the auxiliary parameters by $v \sim p(\cdot|\theta,\rho)$.
    \item Compute the Metropolis-Hastings-Green (MHG) involution proposal (\ref{eq:genGISTInvolution}).
    \item Accept $((\theta',\rho'),v')$ with probability
\begin{equation}
    \min \left(1,\exp(H(\theta,\rho)-H(\theta',\rho'))\frac{p(v'|\theta',\rho')}{p(v|\theta,\rho)}\right).
    \label{eq:GIST-acc-prob}
\end{equation}
Otherwise, keep $((\theta,\rho),v)$. In either case, $v$ is then discarded.
\end{enumerate}

The momentum refresh in the first step preserves the phase-space
target proportional to $\exp(-H)$. Consequently, if
$(\theta,\rho)$ has this distribution, then sampling
$v\sim p(\cdot\mid\theta,\rho)$ gives
$(\theta,\rho,v)$ the augmented target density
$\mu(\theta,\rho,v)$. The Metropolis--Hastings step in steps 3--4
preserves this augmented target. Finally, discarding $v$ after the
accept--reject step recovers the correct phase-space marginal.

The standard HMC method above clearly fits within the above GIST framework (with $\mathbb V$ being a singleton), as does the (time-discretized) randomized HMC \citep{bou2017randomized}, where the number of integration steps $v = L$ has some distribution $p(L)$ on the positive integers that is independent of $(\theta,\rho)$ (and $g$ is the identity).



\subsection{Locally adaptive step size for standard Hamiltonian Monte Carlo} \label{sec:local_adapt_standard_hmc}


The GIST framework \citep{bou2024gist} provides a recipe for incorporating locally adaptive step size procedure. Consequently, this work follows \cite{bou2025within} in introducing the concepts of \emph{macro} and \emph{micro} steps and correcting for reversibility failures in the step size controller using the GIST Metropolis correction. 

For the standard HMC algorithm, one would perform a fixed number $L > 0$ of macro steps. Each of these steps will correspond to an evolution of the system by a macro step size $h$, which will now be fixed so that the integration time from the initial state to the final state is identical in each iteration. For convenience, define also the uniform time grid as $t_j := jh$, $j \in \mathbb Z$. This definition is rather general and useful for more elaborate samplers where one might perform both backward and forward integration steps. For the standard HMC, it is sufficient to assume that $j \in \mathbb N_0$.  

Within the macro step between $t_j$ and $t_{j+1}$, one performs $l_{j, j+1} \in \mathbb N$ micro (leapfrog) steps, each with step size $hl_{j, j+1}^{-1}$.
The leapfrog map evolving the system from time $t_j$ to $t_{j+1}$ is simply given as $\Phi_{hl_{j, j+1}^{-1}}^{l_{j, j+1}}$ and inherits the time reversibility property (conditionally on $l_{j,j+1}$). See Figure 6 of \cite{bou2025within} for an illustration of the macro/micro step distinction.

Composing $L$ such macro steps leads to the $(\theta,\rho)$-proposal $\Phi_{L, 0} \mid_{\mathbf l}$
\begin{equation*} \label{eq:locally_leapfrog_composition_standard_hmc}
    \Phi_{L, 0} \mid_{\mathbf l} := \begin{cases}
        \Phi_{hl_{L-1, L}^{-1}}^{l_{L-1, L}} \circ \dots \circ \Phi_{hl_{0, 1}^{-1}}^{l_{0, 1}}, \text{ if } L >0, \\
        \mathcal F \circ \Phi_{hl_{L, L+1}^{-1}}^{l_{L, L+1}} \circ \dots \circ \Phi_{hl_{-1, 0}^{-1}}^{l_{-1, 0}} \circ \mathcal F , \text{ if } L <0, 
    \end{cases}
\end{equation*}
where $\mathbf l$ is a vector containing all the micro step size counts $l_{j, j+1}$ such that $(\mathbf l)_j = l_{j - 1, j}$, $j = 1, \dots, L$, for $L \geq 1$. A similar definition is also available for $L\leq -1$ by letting $(\mathbf l)_{-j} = l_{j, j+1}$, $j = L, \dots, -1$.  For $L=0$, which may occur when $L$ is randomized, $\Phi_{0,0}\mid_{\mathbf l}$ is the identity map and $\mathbf l$ is empty.
Note that $\mathcal F \circ \Phi_{L, 0} \mid_\mathbf l$ itself is generally not an involution (unless $\mathbf l$ happens to be palindromic). 

The failure to be an involution may still be resolved by embedding this construction into the GIST framework. More specifically, let the collection of auxiliary parameters be $v=\mathbf l$ and $G_{\mathbf l}= \Phi_{L, 0} \mid_\mathbf l$. Furthermore, by choosing the parameter mapping $g_{(\theta,\rho)}(\mathbf l)=\mathcal R(\mathbf l)$, where $\mathcal R$ reverses the order of the argument, it is clear that
\begin{equation}
(\theta',\rho',\mathbf l')=(\mathcal F \circ \Phi_{L, 0} \mid_\mathbf l(\theta,\rho),\mathcal R(\mathbf l))
\label{eq:stepsize-involution}
\end{equation}
is a volume-preserving involution which may be used within the Metropolis-Hastings-Green framework.


To complete the GIST sampler, it is required to specify the parameter distribution $p(\mathbf l|\theta,\rho)$. This distribution should reflect the main objective of adapting the step sizes locally, i.e., choosing the largest leapfrog step size possible such that a given criterion related to the error in the flow approximation is below a prescribed threshold. In practice, $p(\mathbf l|\theta,\rho)$ is constructed in a sequential manner as
\begin{equation}
p(\mathbf l|\theta,\rho) = \prod_{j=1}^L p(l_{j-1,j}|\theta_{j-1},\rho_{j-1}), 
\label{eq:pl-joint-distr}
\end{equation}
where $(\theta_0,\rho_0)=(\theta,\rho)$, $ (\theta_j,\rho_j)= \Phi_{j, 0} \mid_{ (l_{0,1},\dots,l_{j-1,j})}(\theta,\rho),\;j=1,\dots,L$ (and $(\theta_L,\rho_L)=(\theta',\rho')$). Note that the same single macro step law $p(l\mid\theta,\rho)$ appears in every factor of \eqref{eq:pl-joint-distr}. This time-homogeneity property is what makes the ratio in \eqref{eq:GIST_MHG-alpha-general} below factorize over the macro steps.

The exact specifications of the distribution of the number of micro steps in a given macro step $p(l_{j-1,j}|\theta_{j-1},\rho_{j-1})$ will be discussed in detail shortly. Also, notice that due to the sequential nature of (\ref{eq:pl-joint-distr}), steps 2 and 3 in the GIST sampler described in Section \ref{sec:GIST-basics}  are in practice performed concurrently.

Contingent on the specification from (\ref{eq:pl-joint-distr}), the GIST-MHG acceptance probability (\ref{eq:GIST-acc-prob}) for the proposal represented by (\ref{eq:stepsize-involution}) reduces to 
\begin{equation}
\min \left(1,\;\exp(H(\theta,\rho)-H(\theta',\rho'))\prod_{j=1}^L\frac{p(l_{j-1,j}|\theta_j,-\rho_j)}{p(l_{j-1,j}|\theta_{j-1},\rho_{j-1})} \right).
\label{eq:GIST_MHG-alpha-general}
\end{equation}
Note that an ideal micro step number distribution would be ``reversible'' in the sense that the distribution of the number of micro steps would be identical regardless of whether the macro step would be carried out forward or backward in time, i.e., $p(l|\theta,\rho)=p(l|\mathcal F \circ \Phi_{hl^{-1}}^l(\theta,\rho))$ for all $h>0$, $l\in\mathbb N$ and $(\theta,\rho)\in\mathbb R^{2d}$. In practice, finding such exactly reversible  distributions seems infeasible, and instead this paper considers choices of $p(l|\theta,\rho)$ which are only approximately reversible.  

Although the numerically integrated Hamiltonian flow obtained via (\ref{eq:pl-joint-distr}) may be Metropolis-corrected using (\ref{eq:GIST_MHG-alpha-general}) for a very general choice of $p(l|\theta,\rho)$, only $p(l|\theta,\rho)$ with a specific structure will be considered. More specifically, let  $\kappa(\theta, \rho, h, l)$ be a general measure of error associated with a macro step based on $l$ micro steps. Also, denote the desired error tolerance by $\delta$. As an example, $\kappa(\theta, \rho, h, l) = |H(\theta,\rho) - H(\Phi_{h/l}^{l}(\theta,\rho))|$ for the energy error. Finally, define $\mathcal L$ as the smallest (power of 2) number of micro steps fulfilling the error criterion:
\begin{equation}
    \mathcal L(\theta, \rho, h, \delta) = 2^{c^*}, \text{ where } c^*=\min \{c \in \{c_{\min}, c_{\min}+1, \dots  \} \; : \; \kappa(\theta, \rho,  h, 2^c) < \delta  \}. 
    \label{eq:gen_I_def}
\end{equation}
Here, $c_{\min} \in \mathbb N_0$ is the $\log_2$ of the smallest number of micro steps allowed within a macro step according to the prescribed criterion.  

Throughout the theoretical development, all numerical maps are
understood in exact arithmetic. It is assumed that, for every fixed
$h>0$ and $\delta>0$, the mappings
\[
    (\theta,\rho)
    \longmapsto
    \kappa(\theta,\rho,h,2^c)
\]
are measurable and satisfy
\[
    \lim_{c\to\infty}
    \kappa(\theta,\rho,h,2^c)
    =
    0
\]
for almost every $(\theta,\rho)$ with respect to the phase-space
target distribution. If a numerical trajectory is not well-defined
for a particular value of $c$, its error is assigned the value
$+\infty$. These assumptions ensure that the set in
\eqref{eq:gen_I_def} is nonempty almost surely, that
$\mathcal L(\theta,\rho,h,\delta)$ is finite, and that
$\mathcal L$ is measurable.

Two remarks are in order. First, in Algorithm~\ref{alg:alg_hmc} below, $\mathcal L$ is evaluated not only at the initial state, but also at the intermediate states $(\theta_j,\pm\rho_j)$. The map $(\theta,\rho)\mapsto(\theta_j,\pm\rho_j)$ is defined piecewise: on the set of initial states for which a given vector of micro step counts $(l_{0,1},\dots,l_{j-1,j})$ is selected, it is equal to a fixed composition of leapfrog maps and $\mathcal F$, which is a volume preserving bijection of $\mathbb R^{2d}$ and therefore maps Lebesgue null sets to null sets. There are countably many such sets and they are measurable, so the map sends null sets to null sets on all of $\mathbb R^{2d}$. Consequently, $\mathcal L$ is finite at every state at which the algorithm evaluates it, for $\mu$-almost every initial state. On the exceptional null set, the transition may be defined arbitrarily (for instance as a rejection) without affecting invariance.

Second, the limit assumption above holds for the three error measures considered in this paper under mild conditions. Suppose that $\log\tilde\pi\in C^4(\mathbb R^d)$ is bounded above. Then, the Hamiltonian flow $\varphi_t$ is complete since the conservation of $H$ confines each trajectory to a compact set, and for every $(\theta,\rho)$, the leapfrog approximation $\Phi_{h2^{-c}}^{2^c}(\theta,\rho)$ converges to $\varphi_h(\theta,\rho)$ as $c\to\infty$ \citep[Chapter~II]{hairer2008solving}. Consequently, $\kappa_H\to|H-H\circ\varphi_h|=0$; $\kappa_{MRF}\to0$ because both the fine and the coarse leapfrog flows converge to $\varphi_h$; and $\kappa_{RKNF}\to0$ because each of the terms $f_k,b_k$ given by \eqref{eq:flow_criterion_eq_2}-\eqref{eq:flow_criterion_eq_3} is a local error of size $O((h2^{-c})^3)$ uniformly in $k$ by Lemma~\ref{lem:rkn_order} from Appendix \ref{appendix:rkn_order}. 

In the implementation, the search in \eqref{eq:gen_I_def} is truncated at a maximal level $c_{\max}$ (here, $c_{\max}=20$ is used), and $\mathcal L:=2^{c_{\max}}$ when no $c\leq c_{\max}$ satisfies the criterion. This truncated rule is itself a deterministic choice of $\mathcal L$, it removes the need for the limit assumption, and Proposition~\ref{prop:approx_reverse} below continues to hold for it: if $\mathcal L(\theta,\rho,h,\delta)=2^{c^*}$ with $c^*<c_{\max}$, the proof applies verbatim, and if $\mathcal L(\theta,\rho,h,\delta)=2^{c_{\max}}$, the inequality \eqref{eq:I-inequality} is trivial. Accordingly, in the backward check of Algorithm~\ref{alg:macro-step}, the value $l_j^b:=l_j^f$ is assigned whenever no $c<c^*$ satisfies the criterion. None of the experiments reported below reached $c_{\max}$.


In this paper, only a deterministic variant of $p(l|\theta,\rho)$ is considered. This is specified by
\begin{equation}
p(l|\theta,\rho) = \begin{cases}
    1, & \text{if}\; l=\mathcal L(\theta, \rho, h, \delta),\\
    0, & \text{otherwise}.
\end{cases}
\label{eq:pl_deterministic}
\end{equation}
Consequently, this leads to the acceptance probability
\begin{equation} \label{eq:metropolis_adapt_standard_hmc}
    \alpha((\theta_0, \rho_0), (\theta_L, -\rho_L)) = \min\parentheses{1, \exp(H(\theta_0, \rho_0) - H(\theta_L, -\rho_L))} \prod_{j = 1}^L w_j',
\end{equation}
where $w_j' = 1$ if $l_j^f = l_j^b$ and 0 otherwise. Here, $l_j^f := \mathcal L (\theta_{j-1}, \rho_{j-1}, h, \delta)$ and $l_j^b := \mathcal L (\theta_{j}, -\rho_{j}, h, \delta)$. Overall, a single iteration of the locally adaptive step size version of standard HMC as a GIST sampler is summarized in Algorithm \ref{alg:alg_hmc} after omitting the auxiliary variables (except for the momentum) for clarity.   Because the momentum is sampled independently from $N(0,M)$ at every
iteration and discarded after the accept--reject step, Algorithm
\ref{alg:alg_hmc} defines a Markov transition kernel on the position
variable $\theta$ alone.
\begin{algorithm}[t]
\caption{A single iteration of the locally adaptive step size standard HMC algorithm}\label{alg:alg_hmc}
\begin{algorithmic}[1]
\State Predefine the number of macro steps $L$ and the macro step size $h$. 
\State Let $\theta^{(n)}$ denote the current position and set
$\theta_0=\theta^{(n)}$.
Independently sample a new momentum $\rho_0\sim N(0,M)$.
\For{$j = 1,\dots,L$} 
\State \textcolor{blue}{Compute $l_j^f = \mathcal L(\theta_{j-1}, \rho_{j-1},  h, \delta)$.}
\State \textcolor{blue}{Set $(\theta_{j}, \rho_{j}) = \Phi_{h(l_j^f)^{-1}}^{l_j^f} (\theta_{j-1}, \rho_{j-1})$.} 
\State \textcolor{blue}{Compute $l_j^b = \mathcal L(\theta_{j}, -\rho_{j},  h, \delta)$.}
\State \textcolor{blue}{Set $w_j' = 1$ if $l_j^f = l_j^b$ and 0 otherwise.} 
\EndFor
\State With probability given by
\eqref{eq:metropolis_adapt_standard_hmc}, set
$\theta^{(n+1)}=\theta_L$. Otherwise, set
$\theta^{(n+1)}=\theta^{(n)}$. The momentum variables are then
discarded.
\end{algorithmic} 
\end{algorithm}

Before proceeding, a few remarks on the choices made above are discussed. Firstly, it should be noted that the practice of considering only $2^c$, $c=c_{\min},c_{\min}+1,\dots,$ micro steps per macro step is somewhat arbitrary, but has been found to work well in practice \citep{bou2024incorporating,bou2025within}. Different schedules of progressively finer discretization could also be possible, but would have to balance the need for arriving fast enough at an acceptable $l$ (without computationally wasteful search through unnecessary coarser discretizations) while not being too aggressive at the same time (and thereby wasting computation by choosing substantially finer discretizations than what is needed). Secondly, the choice of a deterministic $p(l|\theta,\rho)$ is also made based on experience with computational performance. \cite{bou2024incorporating,bou2025within} also consider randomized variants of (\ref{eq:pl_deterministic}) that place mass on both $\mathcal L$ and the next finer discretization (with equal mass on the two in \cite{bou2025within}), which may preclude outright rejection of the complete proposal in cases where $l^b<l^f$ (see Proposition \ref{prop:approx_reverse} below). However, this practice will also entail that overly fine discretizations are being used (when $l=\mathcal L +1$), and the approach was found by \cite{bou2025within} to be generally less computationally efficient than (\ref{eq:pl_deterministic}).

\subsection{On approximate reversibility}

\begin{figure}
    \centering
    \includegraphics[scale=0.4]{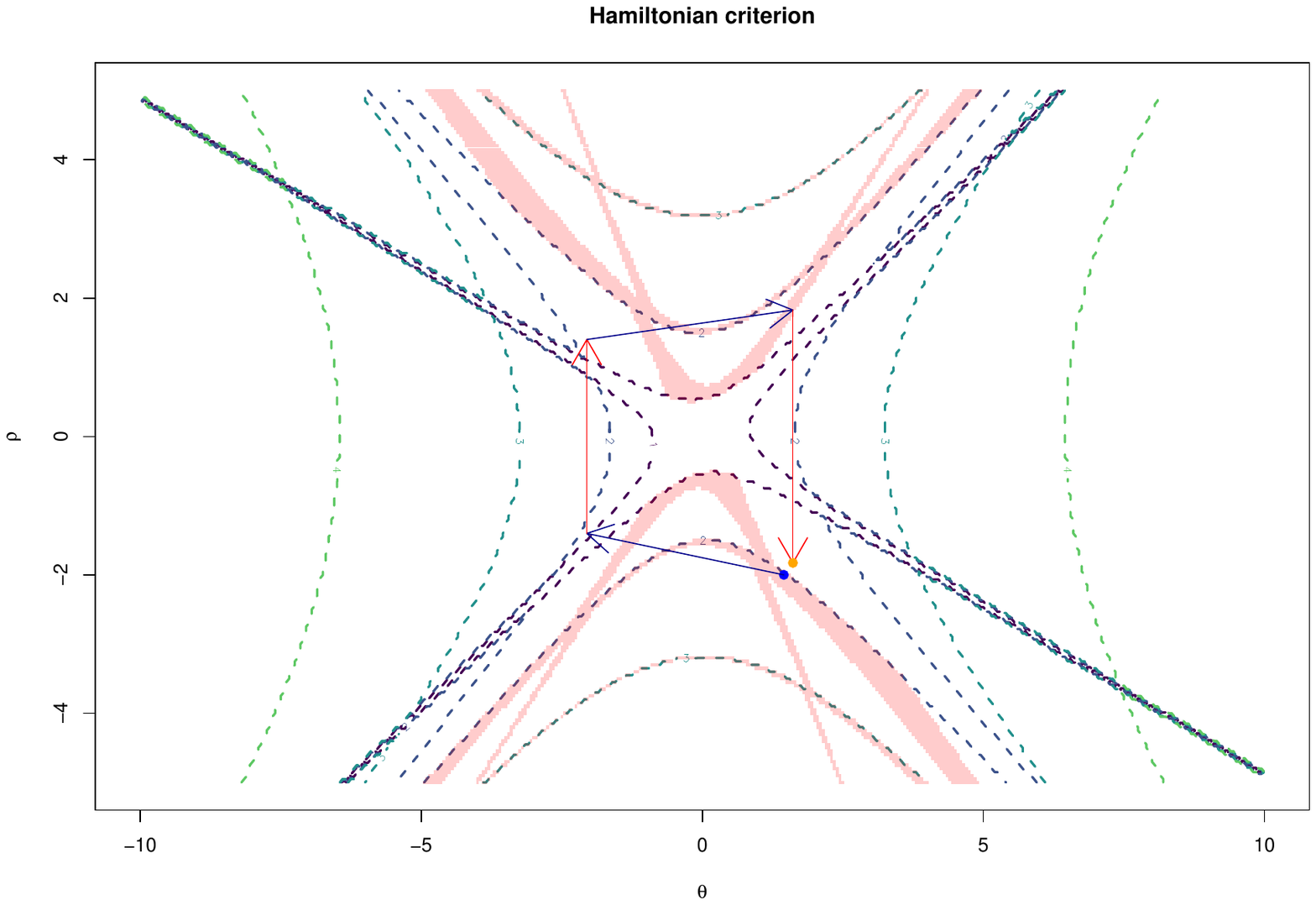}
    \caption{A concrete example of a scenario where reversibility fails for the Hamiltonian criterion. The target distribution is a univariate standard Gaussian, $M=1$, $h=\pi/2$, $c_{\min}=0$ and $\delta=0.2$. The blue point represents an initial state $(\theta, \rho) = (1.45, -2.00)$. The blue arrows denote forward macro steps while the red arrows correspond to momentum flips. The first blue arrow therefore shows the first forward macro step while the remaining constitute the backward macro step with the orange point as the final destination. The red shaded area illustrates the region of points that will experience reversibility failures, the dashed lines are the boundaries between areas of two neighbouring micro step sizes. }
\label{fig:concrete_example_hamiltonian_criterion_1d_normal}
\end{figure} 

    \label{fig:concrete_example_og_hamiltonian_criterion_1d_normal}

Before proceeding, a further notion is useful. An error measure $\kappa$ is said to be reversible if, for all $h>0$, all $l\in\{2^c:c\geq c_{\min}\}$ and all $(\theta,\rho)\in\mathbb R^{2d}$ for which both sides are defined, the following condition is satisfied:
\begin{equation} \label{eq:error_measure_reversible_def}
\kappa(\theta,\rho,h,l)=\kappa(\theta',-\rho',h,l), \; \text{where } (\theta',\rho')=\Phi_{h/l}^l(\theta,\rho).
\end{equation}
For instance, because the leapfrog integrator is exactly reversible
and $H\circ\mathcal F=H$,  it is clear that the absolute energy error $\kappa(\theta, \rho, h, l) = |H(\theta,\rho) -H(\Phi_{h/l}^{l}(\theta,\rho))|$ is an example of a reversible error measure. In the remaining part of the text, only reversible error measures will be considered. For reversible error measures, the following result, which generalizes Lemma 14 of \cite{bou2025within}, holds:

\begin{prop}
\label{prop:approx_reverse}
Suppose that $\kappa$ is a reversible error measure. Then
\begin{equation}
\mathcal L(\theta,\rho,h,\delta) \geq \mathcal L(\theta^*,-\rho^*,h,\delta)\; \text{ where } (\theta^*,\rho^*) = \Phi^{\mathcal L(\theta,\rho,h,\delta)}_{h/\mathcal L(\theta,\rho,h,\delta)}(\theta,\rho). 
\label{eq:I-inequality}
\end{equation}
\end{prop}


\begin{proof}
Let
$\mathcal L(\theta,\rho,h,\delta)=2^{c^*}$,
where $c^*$ is the minimum in \eqref{eq:gen_I_def}. By definition,
\[
\kappa(\theta,\rho,h,2^{c^*})<\delta.
\]
Because $\kappa$ is reversible and
\[
(\theta^*,\rho^*)=\Phi_{h2^{-c^*}}^{2^{c^*}}(\theta,\rho),
\]
it follows that
\[
\kappa(\theta^*,-\rho^*,h,2^{c^*})<\delta.
\]
Thus, $c^*$ belongs to the set over which the minimum defining
$\mathcal L(\theta^*,-\rho^*,h,\delta)$ is taken. Consequently,
\[
\mathcal L(\theta^*,-\rho^*,h,\delta)
\leq 2^{c^*}
=
\mathcal L(\theta,\rho,h,\delta),
\]
which proves the result.
\end{proof}

Reversibility of the error measure $\kappa$ in the sense of \eqref{eq:error_measure_reversible_def} does not make the micro step number distribution \eqref{eq:pl_deterministic} reversible in the sense of Section~\ref{sec:local_adapt_standard_hmc}: Proposition~\ref{prop:approx_reverse} yields only the inequality $\mathcal L(\theta^*,-\rho^*,h,\delta)\leq\mathcal L(\theta,\rho,h,\delta)$. The indicator $w_j'$ in \eqref{eq:metropolis_adapt_standard_hmc} is equal to one when the equality holds in the macro step $j$ and zero when the inequality is strict.

Proposition \ref{prop:approx_reverse} provides several insights. First, it is clear that any reversibility failures originate from the fact that a coarser discretization is acceptable when integrating backward in time. An example of such a reversibility failure for a standard Gaussian target distribution with $c_{\min}=0$ and the energy error measure is presented in Figure \ref{fig:concrete_example_hamiltonian_criterion_1d_normal}. For a fixed combination of $h$ and $\delta$, the regions in the phase space leading to reversibility failures are indicated by the red shaded area. The dashed lines represent the boundaries between different choices of $c=\log_2(l)$. The reversibility failures occur when $(\theta^*,-\rho^*)$ (the state after the first blue arrow originating from the blue dot and the first red arrow) is in a lower $c$-region ($c=1$, $l=2$) than $(\theta,\rho)$ (i.e., the blue dot itself, which is located in the region of $c=2$, $l=4$). As the boundaries between different micro step numbers are determined by the particular error measure $\kappa$ employed, the first motivation of this paper is to investigate and develop other error measures beyond the energy error that are less susceptible to such reversibility failures. Using Figure \ref{fig:concrete_example_hamiltonian_criterion_1d_normal}, this corresponds to reducing the total area (or volume in higher dimensions) of the red shaded regions. 

A second insight provided by Proposition \ref{prop:approx_reverse} is that if $\mathcal L(\theta,\rho,h,\delta)=2^{c_{\min}}$, one necessarily obtains equality in (\ref{eq:I-inequality}) and therefore with certainty achieves a reversible macro step. In the graphical representation in Figure \ref{fig:concrete_example_hamiltonian_criterion_1d_normal}, this observation translates to the absence of red shaded area within the (innermost) region where $c=c_{\min}=0$, $l=1$. This property may be leveraged to reduce the number of reversibility failures, typically by reducing the macro step size $h$ for a fixed error threshold $\delta$ so that the majority of macro steps involve only a single micro step. A value of $P(c=c_{\min})$ around 0.8 has been found to work well in practice \citep{bou2025within}. Such choices of $h$ appear to be a reasonable balance between (for small $h$) certain reversibility, but also the potential for unnecessarily high accuracy/inefficient simulations, against (for large $h$) better adapted micro step sizes, but also higher rates of reversibility failures which may lead to inefficiencies in the overarching MCMC simulation.

\begin{algorithm}[t]
\caption{A single optimized macro step assuming a reversible error criterion $\kappa$ (replaces the macro step $j$ indicated by the blue lines in Algorithm \ref{alg:alg_hmc}). Note that the computation of $\Phi_{h2^{-c}}^{2^c}$ and the associated error measure $\kappa$ is typically done concurrently as indicated below.}\label{alg:macro-step}
\begin{algorithmic}[1]
\For{$c = c_{\min},c_{\min}+1,\dots$} 
    \State Compute $((\theta_j,\rho_j),k)=\left(\Phi_{h2^{-c}}^{2^c}(\theta_{j-1},\rho_{j-1}),\kappa(\theta_{j-1},\rho_{j-1},h,2^c)\right)$
    \If{$k<\delta$}
        \State $c^* =c$
        \State \texttt{break}
    \EndIf
\EndFor
\State $w_j'=1$
\If{$c^*>c_{\min}$}
    \For{$c=c_{\min},c_{\min}+1,\dots,(c^*-1)$}
        \State Compute $k=\kappa(\theta_{j},-\rho_{j},h,2^c)$ (the state $\Phi_{h2^{-c}}^{2^c}(\theta_{j},-\rho_{j})$ is computed as a by-product, but is not needed)
        \If{$k<\delta$}
            \State $w_j'=0$
            \State \texttt{break}
        \EndIf
    \EndFor
\EndIf
\end{algorithmic} 
\end{algorithm}

Algorithm \ref{alg:macro-step} summarizes a single macro step (equivalent to the blue lines in Algorithm \ref{alg:alg_hmc}) which is optimized via the utilization of Proposition \ref{prop:approx_reverse}. In particular, no reversibility check needs to be done if the minimal number of micro steps, i.e., $c^*=c_{\min}$, is acceptable in the forward step. When $c^*>c_{\min}$, the $c=c^*$ backward check is not needed as it is already known from the forward pass to be acceptable due to Proposition \ref{prop:approx_reverse} (since $\kappa$ is assumed to be reversible).

\section{Error measures for the micro step selection criterion} \label{sec:micro_step_criterion}
So far, the error measures for adaptive step size MCMC methods have predominantly been based on energy/Hamiltonian error \citep{SJOS:SJOS12204,pmlr-v238-biron-lattes24a,bou2024incorporating,bou2025within}.
Based on the above discussion, it is clear that one would prefer to avoid using error measures that result in too many $w_j'$ being zero, at least from the perspective of the standard HMC method. Two alternatives that are different from the former will be developed and presented, with the goal of reducing the frequency of reversibility failures. However, such improvements generally come at the cost of additional gradient evaluations per micro step. Whether such added costs pay off in the overall MCMC simulation efficiency will be explored later.

\subsection{Hamiltonian criterion} \label{sec:hamiltonian_criterion}
As presented before, the energy error measure $\kappa_H$ may be summarized by
\begin{equation} \label{eq:standard_hamiltonian_criterion}
    \kappa_H(\theta,\rho,h,l)=|H(\theta,\rho) - H \circ \Phi_{h/l}^{l}(\theta,\rho)|.
\end{equation}
It is straightforward to verify that $\kappa_H$ is reversible, and the additional computational cost of $\kappa_H$ is negligible (since $\Phi_{h/l}^{l}(\theta, \rho)$ must be computed anyway as it corresponds to the final state after a macro step of size $h$). 

For a given combination of the macro step size $h$ and the criterion threshold $\delta$, the $c$-boundaries associated with (\ref{eq:standard_hamiltonian_criterion}) for a univariate standard Gaussian target are illustrated in Figure \ref{fig:concrete_example_hamiltonian_criterion_1d_normal}. The energy error of a single leapfrog step in this case is quadratic in $(\theta,\rho)$. Consequently, the boundary between the region of $c=0$ and $c=1$ is a pair of hyperbolas that both approach asymptotically the roots of the energy error (i.e., the roots of (\ref{eq:standard_hamiltonian_criterion}) in $(\theta,\rho)$). In the case of Figure \ref{fig:concrete_example_hamiltonian_criterion_1d_normal}, these are given as $\rho = (\pi/4) \theta$ and $\rho=[(\pi^2-16)/(4\pi)] \theta$. On the other hand, the error in the position and momentum vector increases linearly in $(\theta,\rho)$. Hence, the energy error criterion (\ref{eq:standard_hamiltonian_criterion}) may in fact allow for an arbitrarily poor representation of the underlying Hamiltonian flow obtained by the leapfrog integrator for an arbitrarily small $\delta$ when the initial $z = (\theta,\rho)$ is close to the roots of the energy error boundary defining the region of one single micro step.

One is not concerned with the accuracy of the representation of the Hamiltonian flow per se (as only the energy error enters the acceptance probability, e.g., (\ref{eq:metropolis_adapt_standard_hmc})). However, tolerating such large flow errors may allow for instabilities in the computation of proposals. Furthermore, and possibly more importantly, it is seen from Figure \ref{fig:concrete_example_hamiltonian_criterion_1d_normal} that the regions of reversibility failures tend (arguably to a varying degree) to be concentrated close to the $c$-boundaries. To utilize this insight in order to achieve fewer reversibility failures, it seems natural to look for error criteria $\kappa$ that induce less complicated $c$-boundaries, and possibly also such that the boundaries are more aligned with the flow (vector) field in the phase space.

\subsection{Embedded integration step flow criterion} \label{sec:flow_criterion}
In this subsection, to obtain an error measure with simpler $c$-boundaries, an alternative criterion based on measuring flow error is developed and presented. More specifically, the embedded integration step approach, commonly used for step size control for Runge-Kutta methods \citep[see e.g.,][]{hairer2008solving}, is employed. Embedded integration steps involve at least two methods for advancing the flow with different error orders, while using at least some of the same force (which corresponds to the gradient of the log-target density here) evaluations. The difference between these methods is interpreted as an approximation to the local error of the method with the lowest error order, and this estimate may subsequently be used for step size control.  

Here, to the authors' knowledge, a novel third-order Runge-Kutta-Nystr\"{o}m (RKN) step \citep{hairer2008solving} is embedded with a leapfrog step for measuring the second-order leapfrog flow error in a reversible manner. The RKN step exploits an already computed leapfrog step and adds a single additional gradient evaluation to obtain a higher order method. In particular, for a general step size $\epsilon$, denote by $(\theta, \rho)$ the initial state and $(\theta', \rho') = \Phi_{\epsilon}(\theta, \rho)$. To get a third-order integrator based on the already computed leapfrog state $(\theta', \rho')$ (and the associated log-target density gradients), define $\Bar{\Phi}_{\epsilon}: (\theta, \rho) \mapsto (\theta'', \rho'')$ in the following manner:
\begin{eqnarray} \label{eq:third_order_integrator_flow}
\theta_m &=& \frac{1}{2}(\theta+\theta') + \frac{\epsilon}{8} M^{-1}(\rho-\rho'), \label{eq:third_order_integrator_flow_nr1} \\
\theta'' &=& \theta + \epsilon M^{-1} \rho + \epsilon^2M^{-1} \left(\frac{1}{3} \nabla_{\theta} \log \Tilde{\pi}(\theta)
 + \frac{1}{6} \nabla_{\theta} \log \Tilde{\pi}(\theta')
 \right), \label{eq:third_order_integrator_flow_nr2} \\
\rho'' &=& \rho + \epsilon \left( \frac{1}{6} \nabla_{\theta} \log \Tilde{\pi}(\theta)
 + \frac{2}{3} \nabla_{\theta} \log \Tilde{\pi}(\theta_m)
 + \frac{1}{6} \nabla_{\theta} \log \Tilde{\pi}(\theta')
 \right). \label{eq:third_order_integrator_flow_nr3}
\end{eqnarray}
Note that the intermediate \emph{$\theta_m$ has been selected so that it coincides in the calculation of both $\bar \Phi_\epsilon(\theta,\rho)$ and $\bar \Phi_\epsilon \circ \mathcal F(\theta',\rho')$}. Therefore, obtaining a reversible/symmetric error measure, say $\max( \parallel (\theta',\rho')- \bar \Phi_\epsilon(\theta,\rho) \parallel,\parallel \mathcal F(\theta,\rho)-\bar \Phi_\epsilon \circ \mathcal F(\theta',\rho')\parallel )$, requires only a single additional gradient evaluation $\nabla_{\theta} \log \Tilde{\pi}(\theta_m)$ (and a handful of vector operations) relative to the underlying leapfrog integration. 

Combining \eqref{eq:third_order_integrator_flow_nr1}-\eqref{eq:third_order_integrator_flow_nr3} and the macro/micro-construction yields the RKN-flow-based (RKNF) error measure
\begin{align}
    \kappa_{RKNF}(\theta,\rho,h,l) &= \max_{k\in \{1,\dots,l\}} \left\{ \max(f_{k},b_{k})\right\}, \; \text{where} \label{eq:flow_criterion_eq_1} \\
    f_{k} &= \parallel (\theta^{(k)}, \rho^{(k)}) - \Bar{\Phi}_{h2^{-c}} (\theta^{(k-1)}, \rho^{(k-1)}) \parallel, \label{eq:flow_criterion_eq_2} \\
    b_{k} &= \parallel  \mathcal F(\theta^{(k-1)}, \rho^{(k-1)}) - \Bar{\Phi}_{h2^{-c}} \circ \mathcal F(\theta^{(k)}, \rho^{(k)})\parallel \label{eq:flow_criterion_eq_3}
\end{align} 
with $(\theta^{(k)},\rho^{(k)})=\Phi_{h/l}^k(\theta,\rho)$ (and $\Phi^0=I$). 
Here, $\parallel \cdot \parallel$ is an arbitrary (vector) norm, which in the numerical computations of this paper will be the $L^\infty$-norm, and $l = 2^c$, where $c \geq c_\text{min}$ for some $c_\text{min} \geq 0$. Note that the leapfrog flow is used for advancing the integrator in accordance with (\ref{eq:stepsize-involution}). The RKNF steps, which are generally not reversible, are only used for the error measurement so that the resulting RKN states are subsequently discarded. 

Due to the measurement of error both for the forward and backward micro steps, the following result holds:
\begin{prop} \label{prop_flow_criterion}
    $\kappa_{RKNF}$ defined by \eqref{eq:flow_criterion_eq_1}-\eqref{eq:flow_criterion_eq_3} is reversible  in the sense of \eqref{eq:error_measure_reversible_def}.
\end{prop}
\begin{proof}
    See Appendix \ref{appendix:proof_prop_flow_criterion}.
\end{proof}
The proof uses no property of $\parallel\cdot\parallel$ beyond it being a norm on $\mathbb R^{2d}$. In particular, the norm does not need to be invariant under $\mathcal F$.
Note that the error tolerance $\delta$ employed in (\ref{eq:gen_I_def}) is no longer interpretable in terms of energy errors when (\ref{eq:gen_I_def}) is used together with \eqref{eq:flow_criterion_eq_1}-\eqref{eq:flow_criterion_eq_3}, and must in general be tuned so that the energy error entering an acceptance probability (e.g., (\ref{eq:metropolis_adapt_standard_hmc})) is in an appropriate range. However, it should be noted that for proposals involving multiple macro steps, tuning of $\delta$ in conjunction with the energy error (\ref{eq:standard_hamiltonian_criterion}) is generally also needed \citep[see][]{bou2025within}. 

The RKN map
\eqref{eq:third_order_integrator_flow_nr1}-\eqref{eq:third_order_integrator_flow_nr3}
is third-order accurate, and a verification is given in
Appendix~\ref{appendix:rkn_order}. It requires only one additional
gradient evaluation relative to leapfrog. However, it is not optimized
for stability, so the RKN comparison map may
dominate the error criterion and force the selection of an unnecessarily 
fine discretization for large micro step sizes.

\subsection{Multi-resolution flow criterion} \label{sec:bs_criterion}

The flow criterion from the previous subsection considers the norm of the difference between the leapfrog and a third-order accurate state for a given micro step, but the cruder leapfrog state is used to advance the flow (and the more accurate and computationally costly RKN step is subsequently discarded). As known in the numerical ODE literature, local extrapolation, whereby the more precise integrator is also used for advancing the integration process, may in many cases be beneficial \citep[][Chapter II.5]{hairer2008solving}. In this subsection, another error measure based on comparing the leapfrog integration flow using two different step sizes  is presented. In contrast to the previous criterion, the following quantity allows for advancing the integration process using the most precise leapfrog flow approximation computed. Superficially, this approach has some similarities to the Bulirsch-Stoer (BS) method \citep{hairer2008solving} in the sense that the same time interval is traversed multiple times using the same integrator, but with different step sizes and number of integration steps. However, BS methods also involve extrapolation steps, which are infeasible in the present context as these may break the reversibility of the resulting integration process.

More specifically, the proposed error measure, which will now be referred to as multi-resolution flow (MRF) error, applies only for $c_{\min}\geq 1$ and may be summarized for $l=2^c,\;c=c_{\min},c_{\min}+1,\dots$, by 
\begin{equation}
\kappa_{MRF}(\theta,\rho,h,l)=\max\left( 
\parallel (\theta',\rho') - \Phi^{l/2}_{2h/l} (\theta,\rho)  \parallel,
\parallel \mathcal F(\theta,\rho) - \Phi^{l/2}_{2h/l} \circ \mathcal F(\theta',\rho')\parallel \right), \text{ where } (\theta',\rho')=\Phi^{l}_{h/l} (\theta,\rho).
\label{eq:bs_criterion}
\end{equation}
In other words, the error estimate associated with traversing the macro step using $l=2^c,\;c=c_{\min},c_{\min}+1,\dots$ micro steps is obtained by comparing with the flow obtained using $2^{c-1}$ micro steps (i.e., with a micro step size that is twice as large) in either direction. Figure \ref{fig:bs_illustration} provides an illustration of the error measurement (\ref{eq:bs_criterion}) for $l=4$. The ``backward checks'' (second argument of the $\max$-function in (\ref{eq:bs_criterion})) ensure the following: 

\begin{prop} \label{prop_bs_criterion}
    For $l=2^c,\; c=c_{\min},c_{\min}+1,\dots$,  $\kappa_{MRF}$ defined in \eqref{eq:bs_criterion} is reversible in the sense of \eqref{eq:error_measure_reversible_def}.
\end{prop}
\begin{proof}
    See Appendix \ref{appendix:proof_prop_bs_criterion}. 
\end{proof}

In the numerical computations presented below, $L^\infty$-norms are again used in (\ref{eq:bs_criterion}). Note that (\ref{eq:bs_criterion}) may be interpreted as approximately measuring the error of traversing the macro step using $l/2$ leapfrog micro steps. As such, one would expect the error tolerance $\delta$ to be larger for MRF than for RKNF to obtain the same accuracy for the selected $(\theta',\rho')$. On the other hand, since $(\theta',\rho')$ is based on at least two leapfrog steps, MRF-based macro step sizes $h$ could also possibly be larger than the RKNF counterparts. 

Nevertheless, since the error estimation for MRF happens relative to a cruder discretization, one would still require macro step sizes that lead to numerical stability when performing $2^{c-1}$ leapfrog steps in order to select $(\theta',\rho')$ based on $2^{c}$ leapfrog steps.
Similarly to RKNF, the requirement for stability of the forward and backward numerical flows with $2^{c_{\min}-1}$ leapfrog steps in order to select $(\theta',\rho')$ based on $2^{c_{\min}}$ leapfrog steps may in certain cases put an upper bound on the integration error attainable and/or make the MRF-based controller insensitive to $\delta$ for large values of $\delta$. Consequently, it is conceivable in certain situations that MRF-based control can lead to unnecessarily precise/inefficient integration relative to the other controllers.

Finally, it should be noted that by careful implementation of $\kappa_{MRF}(\theta,\rho,h,2^c)$ for consecutive values of $c$,  computational savings can be realized relative to each sequential computation of (\ref{eq:bs_criterion}) from scratch as indicated in Algorithm \ref{alg:macro-step} (see also Table \ref{tab:criteria_computational_cost}). These savings originate primarily from the fact that the high-resolution forward flow in $\kappa_{MRF}(\theta,\rho,h,2^{c})$ is identical to the low-resolution forward flow in $\kappa_{MRF}(\theta,\rho,h,2^{c+1})$. See Appendix \ref{appendix:criteria_computational_cost} for more details regarding this matter.

\subsection{Computational cost}

\begin{table}
    \centering
    \begin{tabular}{lccc} \hline
         Criterion & \multicolumn{3}{c}{Number of gradient evaluations} \\
         & forward pass & backward pass & total\\
         \hline
         Hamiltonian ($c \geq 0$) & $2^{c+1}-1$ & $2^c-1$ & $3(2^c) -2$ \\
         RKNF ($c \geq 0$) & $2(2^{c + 1} - 1)$ & $2(2^{c} - 1)$ & $6 (2^c) - 4$ \\
        MRF ($c \geq 1$) & $\frac{3}{2}(2 ^ {c+1}) - 2$ & $2^c - 2$ & $4 (2^c) -4$ \\
        \hline
    \end{tabular}%
    \caption{Computational cost of a full macro step where the final new configuration $(\theta',\rho')$ is computed using $2^c$ micro steps \emph{assuming that $c_\text{min} = 0$ for the Hamiltonian energy and RKNF error criterion and $c_\text{min} = 1$ for MRF error criterion}. Note that for RKNF and $c>0$, the gradient count is an upper bound as the computation of the error measure (\ref{eq:flow_criterion_eq_1}) may be stopped early if one of $\left\{f_{k},b_{k}\right\}>\delta$.}
    \label{tab:criteria_computational_cost}
\end{table}


Under the assumption that a single given macro step is reversible, Table \ref{tab:criteria_computational_cost} provides a summary of the computational cost in terms of gradient evaluations to obtain the final state of the macro step. The new configuration $(\theta',\rho')$ is computed using $2^c$ micro steps. Although MRF and RKNF are more costly relative to the Hamiltonian error measure per se, these costs come with benefits that are explored further below. 

In addition, observe that the computation of $\kappa_{RKNF}$ for $c>0$ may be stopped early if a micro step error estimate 
$f_k$
is found to exceed $\delta$. Accordingly, the results of RKNF in Table \ref{tab:criteria_computational_cost} should be regarded as upper bounds. It is also worth noting that a variant of the energy error measure \citep[see Equation 9 of][]{bou2025within} could also allow for such early stopping. However, in line with Remark 13 of \cite{bou2025within}, only (\ref{eq:standard_hamiltonian_criterion}) is considered in this work.


\begin{figure}
    \centering
    \begin{tikzpicture}[node distance = 2cm, midarrow/.style={postaction={decorate},
                       decoration={markings, mark=at position 0.5 with {\arrow{>}}}}]
          \node[circle, fill=blue, inner sep=2pt, label=above:${(\theta,\rho)}$] (A) at (0,0) {};
          \node[circle, fill=black, inner sep=2pt] (B) at (5,-1) {};
          \node[circle, fill=red, inner sep=2pt, label=above:${\Phi^{l/2}_{2h/l}(\theta,\rho)}$] (C) at (10,-0.5) {};
          \node[circle, fill=black, inner sep=2pt] (D) at (2.5,-2) {};
    \node[circle, fill=black, inner sep=2pt] (E) at (5,-2.25) {};
    \node[circle, fill=black, inner sep=2pt] (F) at (7.5,-2) {};
    \node[circle, fill=orange, inner sep=2pt, label=below:${(\theta',\rho')}$] (G) at (10,-1) {};
    \node[circle, fill=black, inner sep=2pt] (H) at (5,-1.5) {};
    \node[circle, fill=green, inner sep=2pt, label=below:${\mathcal F \circ \Phi^{l/2}_{2h/l}\circ \mathcal F(\theta',\rho')} $] (I) at (0,-1) {};
    
    \draw[midarrow, thick] (A) -- (B);
    \draw[midarrow, thick] (B) -- (C);
    \draw[midarrow, thick] (A) -- (D);
    \draw[midarrow, thick] (D) -- (E);
    \draw[midarrow, thick] (E) -- (F);
    \draw[midarrow, thick] (F) -- (G);
    \draw[midarrow, thick] (G) -- (H);
    \draw[midarrow, thick] (H) -- (I);
    
    \end{tikzpicture}
    \caption{An illustration of the multi-resolution flow (MRF) criterion where the (high-resolution) final state $(\theta',\rho')$ (orange) is the result of $l= 4$ forward micro steps from the initial state $(\theta,\rho)$ (blue). The red state corresponds to the state obtained using $2$ forward micro steps from the initial state. Finally, the green state corresponds to the resulting state after $2$ backward micro steps from $(\theta',\rho')$. The criterion checks if both the norm of the difference between the two "final" states (red and orange) and the norm of the difference between the two "initial" states (blue and green) are less than a desired value $\delta$. Note that the same computations and comparisons would be performed regardless of whether one calculates $\kappa_{MRF}(\theta,\rho,h,4)$ or $\kappa_{MRF}(\theta',-\rho',h,4)$. }
    \label{fig:bs_illustration}
\end{figure}
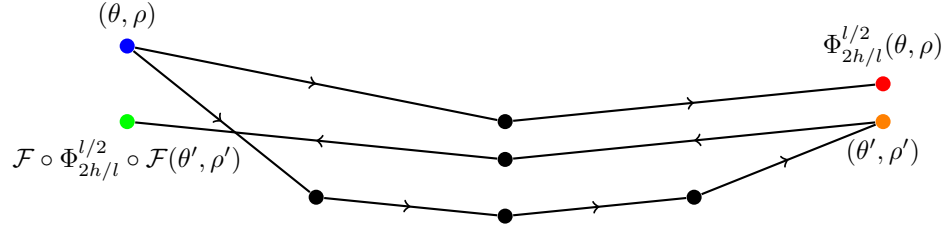


\subsection{Preliminary illustrative examples} \label{sec:prelim_illustrations}
\begin{figure}
    \centering
    \includegraphics[scale=0.45]{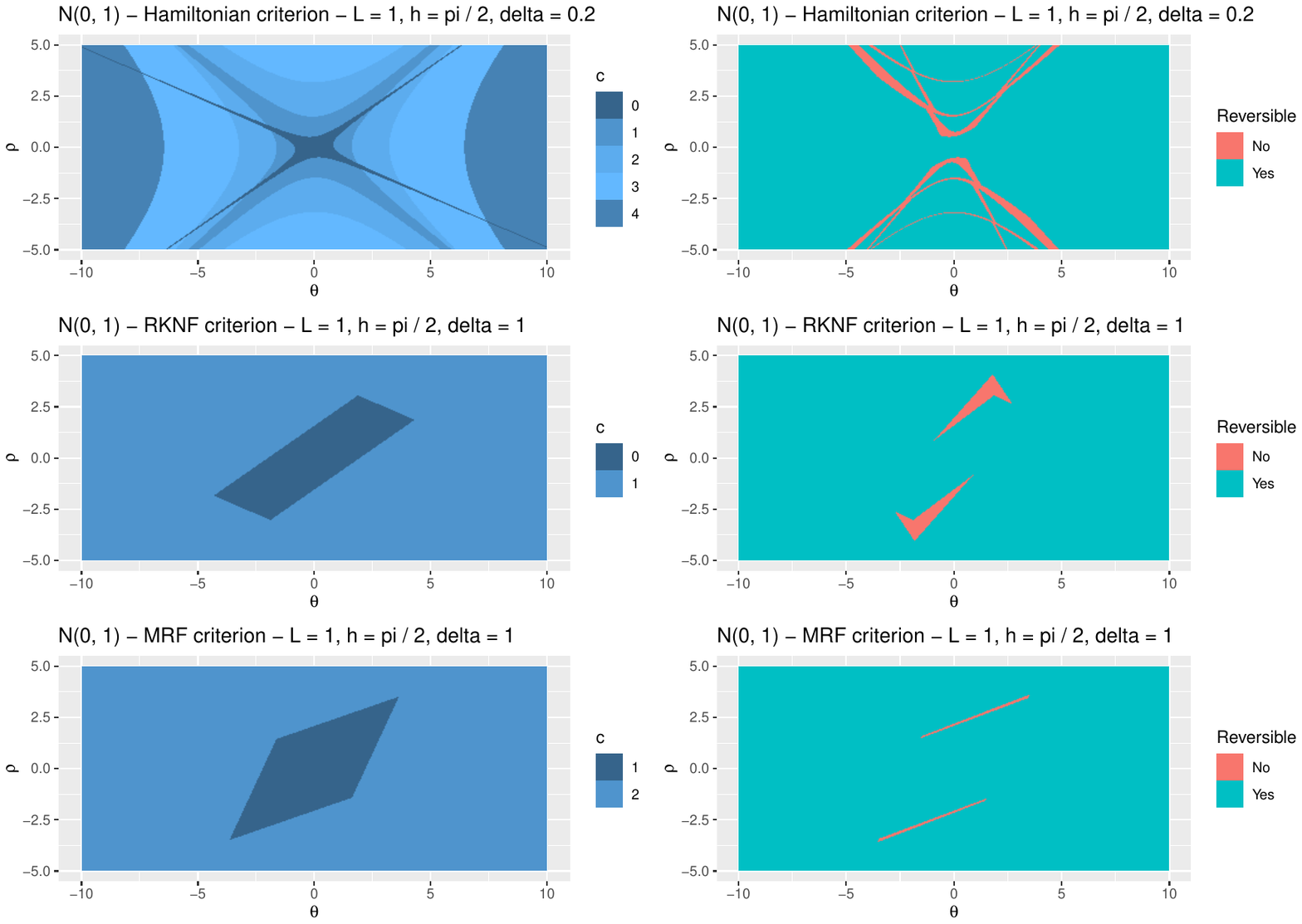}
    \caption{Results for a univariate standard Gaussian and a single integration step. Left plots: $c=\log_2(\mathcal L(\theta,\rho,h, \delta))$ for tuning parameters given in the plots according to a given error measure. Right plots: The red regions are the initial states in the phase space leading to reversibility failures.}
    \label{fig:comparison_1d_normal}
\end{figure} 
The left-hand side panels of Figure \ref{fig:comparison_1d_normal} show the level sets of $c^*$ (see Equation (\ref{eq:gen_I_def})) for a single macro step ($L=1$) with step size control based on the three error measures presented above. The target distribution is a standard Gaussian, where the exact flows are elliptical around the origin. The macro step size is the same in all cases. The choices of $\delta$ approximately lead to around 80\%  acceptance probability (\ref{eq:metropolis_adapt_standard_hmc}) for the Hamiltonian and RKNF criterion. 
For the MRF criterion, due to the $c_{\min} = 1$ floor, the energy acceptance probability from \eqref{eq:metropolis_adapt_standard_hmc} cannot be made smaller than approximately 95\%. 
Also, the flow-based criteria are observed to have less complicated $c$-boundaries, where the boundaries to a larger degree align with the flow of the underlying system.

The right-hand side panels of Figure \ref{fig:comparison_1d_normal} show the initial phase space configurations leading to reversibility failures. Weighted by the Boltzmann-Gibbs distribution $\propto \exp(-H)$, the reversibility failure probability is around  17\% for the Hamiltonian criterion. For RKNF and MRF, the same quantity corresponds to roughly 3\%
and 0.9\%, respectively. 

As such, it is clear that Figure \ref{fig:comparison_1d_normal} may serve as a proof of concept that quite drastic decreases in reversibility failure rates (and thereby increasing the expected number of subsequent reversible macro steps by a sizable factor, say 5-10) are possible with a more careful selection of error measure.

\begin{table}
    \centering
    \resizebox{1\columnwidth}{!}{%
    \begin{tabular}{ccccc}
    \hline
         Criterion & $\delta$ & Reversibility rate &  Metropolis rate & Number of gradient evaluations (in millions) \\
         \toprule
         Hamiltonian & 0.11 & 83.0\% & 80.6\% & 24.7 \\
         RKNF & 0.205 & 97.5\% & 79.8\% & 15.9 \\
        MRF & 1.22 & 97.8\% & 80.3\% & 11.5 \\
        \hline
    \end{tabular}%
    }
    \caption{Reversibility results of the 100-dimensional standard normal distribution using standard HMC with $h = \pi/2$ and $L=1$. }
    \label{tab:100d_normal_reversibility_table}
\end{table}

\begin{table}
\centering
\resizebox{1\columnwidth}{!}{%
\begin{tabular}{cccccccccc}
\toprule
 & \multicolumn{3}{c}{Mean (SD) [Coordinate]} & \multicolumn{3}{c}{{ESS / 100000}} & \multicolumn{3}{c}{$\text{ESS} * 10000 $/ $N_{\text{ode}}$} \\
 \cmidrule(lr){1-1} \cmidrule(lr){2-4} \cmidrule(lr){5-7} \cmidrule(lr){8-10}
 Criterion  & Hamiltonian & RKNF & MRF  & Hamiltonian & RKNF & MRF & Hamiltonian & RKNF & MRF\\
 \cmidrule(lr){1-1} \cmidrule(lr){2-4} \cmidrule(lr){5-7} \cmidrule(lr){8-10}
Parameter &  &  &  &  &  & & & & \\
\midrule
min ESS & 0.0002 (1.0006) [57] & -0.0013 (0.9993) [98] & 0.0008 (0.9999) [39] & 
6.73 & 6.50 & 6.58 & 
273 & 408 & 570 \\
max ESS & -0.0015 (0.9986) [59] & 0.0020 (0.9998) [94] & -0.0015 (0.9997) [84] & 
6.92 & 6.69 & 6.79 & 
281 & 420& 588\\
\bottomrule
\end{tabular}%
}
\caption{ESS results of the 100-dimensional standard normal distribution using standard HMC with $h = \pi/2$ and $L=1$.}
\label{tab:100d_normal_ess_table}
\end{table}

To further explore the effect of varying the error measure, a simple experiment based on a 100-dimensional standard normal target is considered. The macro step size is set to $h = \pi/2$ (which would lead to iid samples if the exact flow is used) with $L = 1$ and full momentum refreshes are performed. For all three criteria, $\delta$ is chosen so that the overall Metropolis acceptance rate is close to 80\%. Ten independent chains are run, each with 100000 iterations in both the warmup and sampling periods. Various diagnostic results for this simple example are summarized in Tables \ref{tab:100d_normal_reversibility_table} and \ref{tab:100d_normal_ess_table}.

Regarding reversibility, one can observe from Table \ref{tab:100d_normal_reversibility_table} that the ratio of macro steps with no reversibility issues is roughly 83\% for the Hamiltonian criterion, which is not far from the target Metropolis acceptance rate of $80\%$. In order to maintain the desired target acceptance rate, smaller micro step size is needed to compensate for this fact, thus leading to more gradient evaluations. For the RKNF and MRF criteria, the reversibility rate is 97.5\% and 97.8\%, respectively. Larger micro step sizes are therefore possible while retaining the desired Metropolis acceptance rate and high reversibility rate. Thus, fewer gradient evaluations are performed in such scenarios, leading to a higher ESS produced per gradient evaluation compared to the standard Hamiltonian criterion based on Table \ref{tab:100d_normal_ess_table}. 

The examples so far illustrate how large an impact the choice of micro step size criterion can potentially have, even for $L = 1$, in very simple situations.  The reversibility problem will become even more apparent for HMC when $L > 1$ as several consecutive macro steps need to be reversible in order not to outright reject the proposal.

\subsubsection{Effect of tuning parameters
  \texorpdfstring{$h$ and $\delta$}{h and delta}}
  \label{sec:tuning_parameter_effect}
  
  \begin{figure}
    \centering
    \includegraphics[width=1\columnwidth]{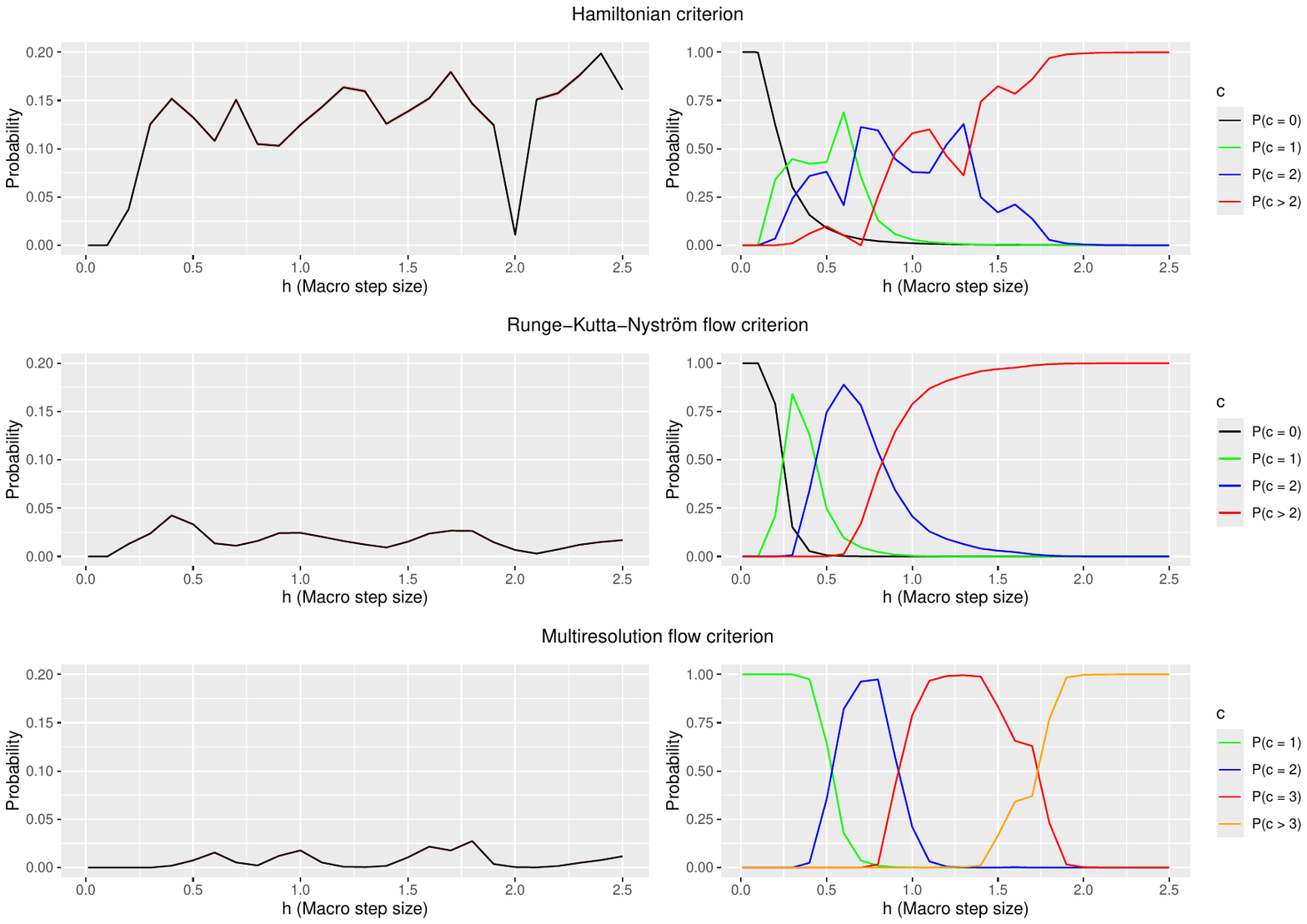}
    \caption{Illustrations of 100000 simulated macro steps with different initial states using the three criteria for the highly correlated bivariate normal distribution. Left plots: Probability of a macro step experiencing the reversibility issue with respect to the macro step size $h$. Right plots: Probability of taking a given number of micro steps to achieve the desired error for a given criterion with respect to the macro step size $h$. Here, $\delta_\text{H} = 0.1$, $\delta_\text{RKNF} = 0.05$ and $\delta_\text{MRF} = 1$. A 95\% confidence interval calculated using the Clopper-Pearson interval is plotted for the plots on the left as red bands, but the width is very small. }
    \label{fig:varying_h_example}
\end{figure} 

\begin{figure}
    \centering
    \includegraphics[width=1\columnwidth]{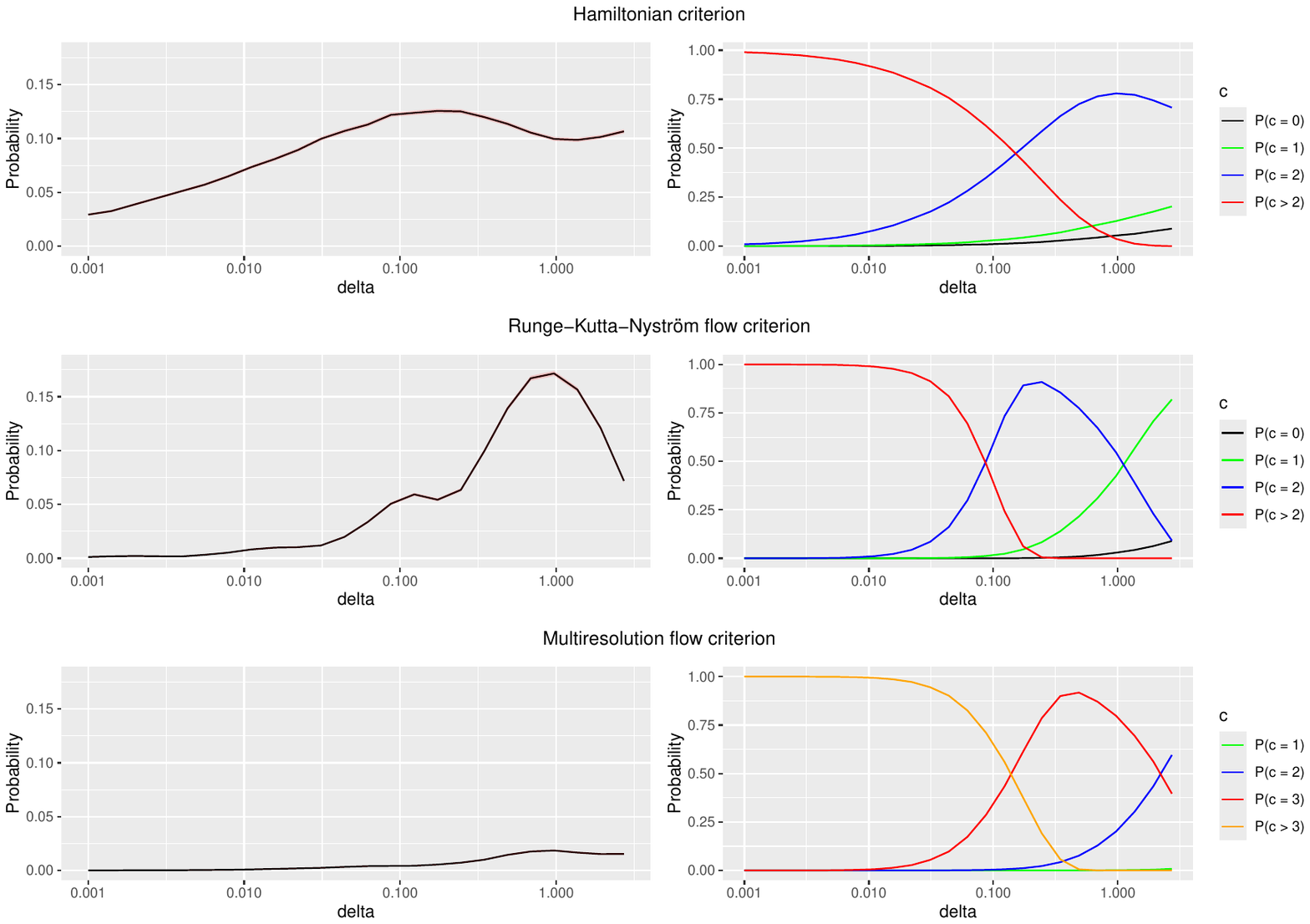}
    \caption{Illustrations of 100000 simulated macro steps with different initial states using the three criteria for the highly correlated bivariate normal distribution. The horizontal axis is plotted in log-scale. Left plots: Probability of a macro step experiencing the reversibility issue with respect to the threshold $\delta$. Right plots: Probability of taking a given number of micro steps to achieve the desired criteria with respect to the threshold $\delta$. Here, $h = 1$. A 95\% confidence interval calculated using the Clopper-Pearson interval is plotted for the plots on the left as red bands, but the width is very small.}
    \label{fig:varying_delta_example}
\end{figure} 

The micro step number selection mechanism given by (\ref{eq:gen_I_def}) involves two tuning parameters; the macro step size $h$ and the error threshold $\delta$. To gain further insight on how these influence the behaviour of the macro step in the context of the three different error measures described above, a set of preliminary and smaller simulation experiments is considered.

First, a bivariate normal distribution with zero mean, unit marginal variances and correlation set to $0.95$ is considered as the target distribution. 
For each criterion, Figure \ref{fig:varying_h_example} shows how the (simulated) reversibility failure probability (left) and the (simulated) distribution of the number of micro steps (right) given a fixed $\delta$ ($\delta=0.1$, $c_{\min}=0$ for $\kappa_H$, $\delta=0.05$, $c_{\min}=0$ for $\kappa_{RKNF}$ and $\delta=1.0$, $c_{\min}=1$ for $\kappa_{MRF}$) vary with respect to $h$. At each value of $h$, 100000 macro steps are performed with initial states drawn according to the Boltzmann-Gibbs distribution $\mu(\theta,\rho)$. 

From Figure \ref{fig:varying_h_example}, one may first observe that $\mathcal L=2^{c_{\min}}$, and therefore a macro step with guaranteed reversibility in line with Proposition \ref{prop:approx_reverse}, can be obtained for each of the error measures by selecting a sufficiently small macro step size. For the selected values of $\delta$, it is seen that $\kappa_H$ leads to a substantially higher prevalence of reversibility failures for larger macro step sizes than the flow-based error measures. Furthermore, Figure \ref{fig:varying_h_example} shows, possibly somewhat counterintuitively, that reversibility failure probabilities are not monotone with respect to the macro step size $h$. Also, presumably due to the rather complicated $c$-boundaries associated with the energy error measure, the distributions of selected micro step numbers (right panel) exhibit quite complicated structure, whereas the flow-based measures transition more smoothly towards more micro steps as $h$ grows.

In Figure \ref{fig:varying_delta_example}, reversibility failure probabilities and micro step number distributions are plotted for the same target distribution, but with fixed $h=1$ and varying $\delta$. In this particular example, except for very low error tolerances (which would lead to highly inefficient proposal generations), it is observed that the energy error criterion leads to quite high reversibility failure probabilities. For RKNF and relevant values of $\delta$, say e.g., below $0.1 \sqrt{\text{Var}(\theta_i)}$, the reversibility failure rate is lower than for the Hamiltonian energy criterion. On the other hand, the MRF reversibility failure probabilities seem to be less sensitive to the choice of $\delta$, at least for the range of values of $\delta$ considered here. Again, note that the reversibility failure probability is not monotone with respect to $\delta$. 

The above discussion provides some insights regarding reasonable tuning strategies for $\delta$ and $h$. As discussed in \cite{bou2025within}, a possible strategy, which will also be pursued to some degree here, is to tune:
\begin{itemize}
    \item $\delta$ towards a desired precision in the flow approximation to control the energy error appearing in the acceptance probabilities (e.g., Equation (\ref{eq:metropolis_adapt_standard_hmc})) or multinomial weights as in the case of WALNUTS.
    \item $h$ so that a sizable proportion of the macro steps that are reachable result in $c=c_{\min}$ (and thereby guaranteed reversibility). In the following, this is simply denoted as $P(c = c_{\min})$. 
\end{itemize}
However, it should be noted that other alternatives also exist. For instance, for the Hamiltonian criterion, the reference implementation of WALNUTS \citep{WALNUTS_cpp} uses fixed $\delta$ and varies $h$ and $c_{\min}$ to obtain efficient sampling. A full discussion of tuning adaptive step size methods is beyond the scope of this paper since the strategy for tuning necessarily also depends strongly on the overarching MCMC method and the error measure employed, as highlighted in Section \ref{sec:numerical_ill}. The choices regarding the chosen tuning parameters leading to the numerical results presented in the remainder of this text are made in the interest of illustrating differences between the criteria and their application in different MCMC methods in various situations. As such, the results should not be interpreted as providing a more formal ranking of the relative merits of the proposed methods, as this would require more tailored tuning strategies for each particular situation.

\subsection{Further remarks on locally adaptive step sizes}
\label{sec:further_remarks}
Before proceeding to examine how locally adaptive step sizes as discussed above may be integrated in various other Hamiltonian flow-based MCMC methods, a few remarks are in order. 

The macro/micro step construction described, as mentioned above, is not restricted to the (fixed mass) Hamiltonian flow. Similar constructions in conjunction with explicit reversible integrators can be made for e.g., Riemannian manifold HMC based on explicit integrators \citep[see e.g.,][]{pmlr-v151-hartmann22a}, various forms of non-volume-preserving dynamics such as the isokinetic/microcanonical flow \citep[see e.g.,][]{JMLR:v24:22-1450,2503.01707} and extended state space methods such as the Nos\'{e}-Hoover thermostat \citep[see e.g.,][]{Leimkuhler_2009}. In particular, if some form of conserved quantity exists, error measures similar to $\kappa_H$ may be utilized. In any case, error measures similar to $\kappa_{MRF}$ will be available with otherwise minimal changes to the above theory.

For reversible implicit integrators, e.g., the generalized leapfrog \citep[see e.g.,][]{Leimkuhler:2004,girolami_calderhead_11} or implicit midpoint, the macro/micro construction may be employed to monitor the convergence of iterative processes such as fixed point iterations used to advance the integrator instead of focusing directly on the flow or energy error quality. However, it should be noted that whether such an iterative process converges in a prescribed number of iterations is generally not reversible in the sense of (\ref{eq:error_measure_reversible_def}).


As briefly mentioned above, the reference WALNUTS C++ implementation \citep{WALNUTS_cpp} increases $c_{\min}$ (and increases $h$) in order to control the number of macro steps/reversibility checks per orbit. This allows for building longer orbits uninterrupted by reversibility failures. However, this practice could introduce computational inefficiencies as the locally adaptive nature of the integrator is to some degree lost, i.e., longer time intervals of Hamiltonian trajectory would need to be integrated with the same step size, even if only a short time interval within the macro step needs a reduction in the (micro) step sizes. A motivation for the present work is to reduce the need for such reduced adaptivity by increasing the reversibility probability.

\section{Locally adaptive step size for other HMC-type methods}\label{sec:other_HMC_methods}

The previous section utilizes the standard HMC method to illustrate how the locally adaptive step size technique may be embedded into an overarching MCMC method. In this section, such embedding is also discussed for other Hamiltonian flow-based MCMC methods with modest modifications. 

It is worth mentioning that enhanced variants of standard HMC such as Randomized HMC \citep{bou2017randomized}, where $L$ is drawn from some distribution $f_L(L)$ (not depending on $(\theta,\rho)$) in each MCMC iteration, can be used within Algorithm \ref{alg:alg_hmc} without altering the acceptance probability. To see this, observe that if $L$ is added to $v \in \mathbb V$ with the associated identity parameter mapping $g$, the factor arising from $f_L$ will cancel in the definition of the acceptance probability (\ref{eq:GIST-acc-prob}). Also, the same applies when step size jittering \citep[see e.g.,][Section 3.3]{neal2011mcmc} is employed so that $h$ is drawn from some distribution (not depending on $(\theta,\rho)$). In that case, for each MCMC iteration, one can rely on the same (random) $h$ for each of the $L$ macro steps or introduce individual (random) values of $h$ in each of the $L$ macro steps. Then, a similar argument to the situation of randomized $L$ applies with the modification that for the case of individual values of $h$, the parameter mapping $g$ must reverse the order of the array corresponding to the values of $h$ following the order of the macro steps (see e.g., Section \ref{sec:GIST-basics}).

\subsection{Langevin/diffusion-based methods}

\subsubsection{MALA}

As explained in \citet[][Section 5.2]{neal2011mcmc}, the Metropolis adjusted Langevin Algorithm (MALA) is equivalent to the standard HMC with $L=1$ and a full momentum refresh \citep{RSSB:RSSB123}. However, unless $c = 0$ for the given macro step (when $\kappa$ is the Hamiltonian or the RKNF criterion), Algorithm \ref{alg:alg_hmc} with $L=1$ is interpreted as simply HMC with $2^c$ leapfrog steps with step size $h / 2^c$. Therefore, note that Algorithm \ref{alg:alg_hmc} with $L=1$ is different from earlier time-step adaptive MALA variants such as \citet{SJOS:SJOS12204} or \citet{pmlr-v238-biron-lattes24a} because the macro/micro construction is not used there, see Remark 16 of \citet{bou2025within}.

\subsubsection{Kinetic Langevin}

Ignoring an adaptive step size scheme, it is well known that MALA may suffer from random walk-like behaviour, in particular for high-dimensional settings. Kinetic/Generalized/second-order Langevin dynamics \citep[see e.g.,][]{Leimkuhler2015molecular} introduces a persistent momentum variable to counteract such behaviour. More specifically, consider the stochastic differential equation
\begin{eqnarray}
    d\theta_t &=& M^{-1}\rho_t dt, \label{eq:kinLangSDE1}\\ 
    d \rho_t &=& \left[ \nabla_\theta \log\tilde{\pi}(\theta_t) - \gamma \rho_t \right] dt+ \sqrt{2\gamma} M^{1/2}dB_t, \label{eq:kinLangSDE2}
\end{eqnarray}
where $B_t$ is a $d$-dimensional standard Brownian motion, $\gamma>0$ is a scalar friction coefficient and $M^{1/2}$ is some matrix square root of $M$. The process defined by (\ref{eq:kinLangSDE1})-(\ref{eq:kinLangSDE2}) admits $\propto \exp(-H)$ as the unique stationary distribution under weak regularity conditions on $\tilde{\pi}(\theta)$. 
Consequently,  time-discretized versions of (\ref{eq:kinLangSDE1})-(\ref{eq:kinLangSDE2}) may serve as proposal mechanisms for MCMC methods with \citep[see e.g., ][]{horowitz1991generalized,pmlr-v151-hoffman22a} or without (see e.g., \citet{Leimkuhler2015molecular} and references therein) subsequent Metropolis corrections.

The approach of using an OBAB discretization with time step $h$
\citep{Leimkuhler2015molecular}, together with a Metropolis correction
of the leapfrog step
\eqref{eq:leapfrog_update_theta}--\eqref{eq:leapfrog_update_rho},
is commonly referred to as generalized HMC (GHMC). Unlike standard
HMC, GHMC retains the momentum between iterations and therefore
defines a Markov chain on phase space.

More precisely, let $(\theta^{(n)},\rho^{(n)})$ denote the current
phase-space state. First, partially refresh the momentum according to
\begin{equation}
    \widetilde{\rho}
    \sim
    N\left(
        \exp(-\gamma h)\rho^{(n)},
        \left(1-\exp(-2\gamma h)\right)M
    \right).
    \label{eq:kinLang_refresh}
\end{equation}
This update preserves the $N(0,M)$ momentum distribution. Set
$(\theta_0,\rho_0)=(\theta^{(n)},\widetilde{\rho})$, and perform one
locally adaptive macro step by computing
\[
    l^f=\mathcal L(\theta_0,\rho_0,h,\delta),
    \qquad
    (\theta_1,\rho_1)
    =
    \Phi_{h/l^f}^{l^f}(\theta_0,\rho_0),
\]
followed by
\[
    l^b=\mathcal L(\theta_1,-\rho_1,h,\delta).
\]
The proposal is accepted with probability
\[
    \alpha
    =
    \mathbf 1_{\{l^f=l^b\}}
    \min\left\{
        1,
        \exp\left(
            H(\theta_0,\rho_0)-H(\theta_1,\rho_1)
        \right)
    \right\}.
\]
If the proposal is accepted, set
\[
    (\theta^{(n+1)},\rho^{(n+1)})
    =
    (\theta_1,\rho_1);
\]
otherwise, set
\[
    (\theta^{(n+1)},\rho^{(n+1)})
    =
    (\theta_0,-\rho_0).
\]
Thus, the momentum is retained after each iteration, with a momentum
flip upon rejection. This construction leaves the phase-space target $\mu\propto\exp(-H)$ invariant: one iteration is the composition of the momentum refresh \eqref{eq:kinLang_refresh}, which is an Ornstein--Uhlenbeck step preserving $N(0,M)$, and hence $\mu$, with the locally adaptive GIST Metropolis step with $L=1$, which is $\mu$-reversible by Section~\ref{sec:local_adapt_standard_hmc}, followed by the momentum flip $\mathcal F$ on rejection; the last two together define a $\mu$-invariant, though not $\mu$-reversible, kernel exactly as in \citet{horowitz1991generalized}. In contrast to
Algorithm~\ref{alg:alg_hmc}, the position sequence alone is generally
not Markov because the momentum is only partially refreshed.

Finally, it should also be mentioned that the application of the macro/micro construction to diffusion processes where the time-discretization does not involve uninterrupted leapfrog steps (i.e., BAB) such as BAOAB \citep[see e.g.,][]{Leimkuhler2015molecular} is presently being investigated.

\subsection{Multinomial HMC methods, biased progressive HMC and WALNUTS}
\subsubsection{Background}

Multinomial HMC methods, in a broad sense encompassing biased progressive HMC, NUTS and WALNUTS, are characterized by the generation of proposals in two stages. Firstly, an \emph{orbit}, consisting of the current state and one or more new candidate states, is constructed. Denote by $\{(\theta_{[k]},\rho_{[k]})\}_{k=a}^b$, $a=0,-1,-2,\dots$ and $b=0,1,2,\dots$, the set of states constituting the orbit, with the convention that $(\theta_{[0]},\rho_{[0]})$ is the current state of the MCMC chain. Fixed step size implementations of multinomial HMC methods typically generate new candidates (separated by $h$ units of time) by sequentially applying $R\geq 1$ leapfrog steps both forward and backward in time, say $(\theta_{[k]},\rho_{[k]}) = \Phi^R_{h/R}(\theta_{[k-1]},\rho_{[k-1]}),\;k=1,\dots,b$ if $b\geq 1$ and $(\theta_{[k]},\rho_{[k]}) = \mathcal F \circ \Phi^R_{h/R} \circ \mathcal F (\theta_{[k+1]},\rho_{[k+1]}),\;k=-1,\dots,a$ if $a\leq -1$.


Secondly, a new state of the MCMC chain is drawn from the states constituting the orbit with multinomial weights $W_{[k]} \propto \exp(-H(\theta_{[k]},\rho_{[k]}))$. In general, there is no accept/reject step, but since the current state $(\theta_{[0]},\rho_{[0]})$ is a part of the orbit, it may be selected in the second step, causing the overall MCMC to be stuck for an iteration. As such, multinomial HMC methods are still reliant on relatively accurate simulation of the underlying Hamiltonian flow in order to keep the weights $\{W_{[k]}\}_{k=a}^b$ from deviating too much from being uniform.  

The different multinomial HMC methods deviate in how the orbit is constructed and also in how the multinomial sampling is performed. A full discussion of these methods is beyond the scope of this text, and the reader is referred to \cite{hoffman2014no,betancourt2017conceptual,bou2024gist,bou2025within} for technical details and proofs of correctness. 
In short, the conventional multinomial HMC approach builds orbits with a fixed number of states $L^* := L + 1 = b - a + 1$, where $L$ is the number of leapfrog steps as before. Given $L$, $b$ is then sampled uniformly from $0, \dots, L$ before setting $a = b - L$. On the other hand, biased progressive HMC and (WAL)NUTS build orbits by doubling the number of states of the orbit either forward or backward in time, i.e., $(a,b) \rightarrow (a,b + L^*)$ or $(a,b) \rightarrow (a-L^*,b)$, with equal probability. This doubling process is performed (at most) $m$ times, where $m$ is the prescribed number of doublings, leading to orbits consisting of $2^m$ states. 

The doubling process allows the (WAL)NUTS orbit building process to stop early if a U-turn criterion is satisfied, thereby saving computations relative to obtaining the $2^m-1$ new states as biased progressive HMC does. Furthermore, provided that the multinomial weights $\{W_{[k]}\}_{k=a}^b$ are close to being uniform, the doubling construction allows biased progressive HMC and most implementations of (WAL)NUTS to sample the new state from the states originating from the last doubling with high probability, thus increasing the probability that the new MCMC chain state is more distant (in $|k|$) compared to the independent multinomial sampling, where the final index defining the next state in the chain is known to have a discrete triangular distribution centred at zero as the marginal distribution \citep{bou2025within}.   

Common to the three above mentioned methods is that they ensure that the orbit building process itself is reversible, and the application of step size adaptivity as considered here would require this reversibility property to be retained. More precisely, for reversibility, it is necessary that for any state $(\theta_{[j]},\rho_{[j]}), \; j \in \{a,a+1,\dots,b\}$ in the original orbit $\{(\theta_{[k]},\rho_{[k]})\}_{k=a}^b$, there exists a corresponding sequence of orbit building decisions such that a (hypothetical) orbit $\{(\theta_{[k]}^*,\rho_{[k]}^*)\}_{k=a^*}^{b^*}$ with current state being $(\theta_{[0]}^*,\rho_{[0]}^*) = \mathcal F (\theta_{[j]},\rho_{[j]})$ contains exactly the momentum flipped states of the original orbit. In other words, for any $j \in \{a,a+1,\dots,b\}$, one has $a^*=j-b,\;b^*=j-a$ and $(\theta_{[j-k]}^*,\rho_{[j-k]}^*) = \mathcal F (\theta_{[k]},\rho_{[k]}), \; \forall \; k\in \{a,a+1,\dots,b\}$. Moreover, the reversibility requires that the probability of the forward and backward orbit building decisions is the same, which is fulfilled both for the conventional multinomial HMC and the doubling construction used in (WAL)NUTS. See the proof in Appendix \ref{appendix:multinomial-correctness} for the multinomial HMC case. The corresponding proof for biased progressive sampling follows closely the proof in \cite{bou2025within}.

\subsubsection{Locally adaptive step sizes for multinomial HMC}

Introducing locally adaptive step sizes in multinomial HMC methods involves first building the orbit using the macro step construction, i.e., 
\begin{equation*}
    (\theta_{[k]},\rho_{[k]}) = \Phi_{k,0}\mid_\mathbf l(\theta_{[0]},\rho_{[0]}),\;\forall  \; k \in \{a,a+1,\dots,-1,1,2,\dots,b\},
    \label{eq:step_size_orbit}
\end{equation*}
where the micro step counts are determined (in the case of the deterministic step size control given by (\ref{eq:pl_deterministic})) as $l_{k-1,k}=\mathcal L(\theta_{[k-1]},\rho_{[k-1]},h,\delta)$, $k=1,2,\dots,b$, and $l_{k,k+1}=\mathcal L(\mathcal F(\theta_{[k+1]},\rho_{[k+1]}),h,\delta)$, $k=-1,-2,\dots,a$. Both the leftmost index $a$ and the rightmost index $b$ of the orbit are decided in the same manner as for the fixed step size scenario indicated above.

Correction for possible reversibility failures is accomplished by modifying the multinomial weights used for sampling the next state $j$. More specifically, define the reversibility indicators
\begin{equation*}
w'_{[k-1,k]} = \begin{cases} 1, & \text{ if } \mathcal L(\theta_{[k-1]},\rho_{[k-1]},h,\delta)=\mathcal L(\mathcal F(\theta_{[k]},\rho_{[k]}),h,\delta),\\
0, & \text{ otherwise},
\end{cases}
\end{equation*}
with $k=a+1,\dots,b$
and the reversibility weights
\begin{equation}\label{eq:reversibility_weights}
    r_{[k]}= \begin{cases} 1, & \text{ if } k=0,\\
    \prod_{n=1}^k w'_{[n-1,n]}, & \text{ if } k\geq 1,\\
    \prod_{n=k+1}^0 w'_{[n-1,n]}, & \text{ if } k\leq -1.
    \end{cases}
\end{equation}
For conventional multinomial HMC, conditional on the constructed
orbit, the next-state index is sampled according to
\begin{equation}
    \mathbb P(J=k\mid\mathcal O)
    =
    \frac{
        r_{[k]}\exp\bigl(-H(\theta_{[k]},\rho_{[k]})\bigr)
    }{
        \displaystyle
        \sum_{j=a}^{b}
        r_{[j]}\exp\bigl(-H(\theta_{[j]},\rho_{[j]})\bigr)
    }.
    \label{eq:adaptive-multinomial-weights}
\end{equation}
Thus, states separated from the current state by a macro step that
fails the reversibility check receive zero probability.

For biased progressive HMC and WALNUTS, the quantities
\[
    r_{[k]}\exp\bigl(-H(\theta_{[k]},\rho_{[k]})\bigr)
\]
are instead used as the state weights within the recursive biased progressive selection rule of \citet[Section~3.3]{bou2025within}. The resulting selection distribution is not, in general, the multinomial distribution \eqref{eq:adaptive-multinomial-weights}, and correctness follows from the reversibility argument of \citet[Section~3.4]{bou2025within}, which will be indicated now. That argument uses three ingredients: (i) the map that re-indexes an orbit so that one of its states becomes the origin, together with the reversed sequence of doubling decisions, is a measure-preserving involution (their Lemma~17); (ii) the doubling decisions are symmetric, so that the original and the re-indexed orbit have the same probability; and (iii) the index selection kernel of biased progressive sampling satisfies an index detailed balance identity (their Lemma~6 in the fixed step size case, and Lemma~21 for WALNUTS) whose proof uses only that the state weights are invariant under re-indexing. With fixed step sizes the weights are $\exp(-H(\theta_{[k]},\rho_{[k]}))$ and this invariance is $H\circ\mathcal F=H$. With locally adaptive step sizes and weights $r_{[k]}\exp(-H(\theta_{[k]},\rho_{[k]}))$, ingredient (i) holds on the set $\{r_{[j]}=1\}$ by Lemma~\ref{lem:reindexing} and the change of variables in the proof of Proposition~\ref{prop:multinomial_correct} in Appendix~\ref{appendix:multinomial-correctness}; ingredient (ii) is unchanged; and the invariance required in (iii) is exactly the identity $r^*_{[j-k]}=r_{[k]}$ of Lemma~\ref{lem:reindexing}. For WALNUTS, the U-turn conditions are evaluated on the same states in both orbits (up to $\mathcal F$ and index reversal, under which the U-turn criterion of \citet{bou2025within} is invariant), so the stopping rule is symmetric as well. The routine verification is therefore omitted and the conclusion is now recorded.

\begin{prop} \label{prop:bp_correct}
The locally adaptive biased progressive HMC and WALNUTS transitions, with state weights $r_{[k]}\exp(-H(\theta_{[k]},\rho_{[k]}))$ and orbits built as described above, are $\pi$-reversible Markov kernels on $\mathbb R^d$.
\end{prop}
\begin{proof}
Lemmas~17 and 21 and Theorem~22 of \citet{bou2025within} apply verbatim with $\exp(-H(\theta_{[k]},\rho_{[k]}))$ replaced by $r_{[k]}\exp(-H(\theta_{[k]},\rho_{[k]}))$, the only property of the weights used there being their invariance under re-indexing, which is Lemma~\ref{lem:reindexing}.
\end{proof}
A self-contained proof for the conventional multinomial case is given in Appendix~\ref{appendix:multinomial-correctness}.


The above observation is in contrast to locally adaptive step size procedures for the standard HMC counterpart, where the whole simulated Hamiltonian trajectory $\Phi_{L,0}\mid_\mathbf l$ must be rejected if a reversibility failure is encountered. One would therefore expect locally adaptive standard HMC to be more sensitive to reversibility failures, thus benefiting more from the improved error measures than multinomial HMC methods. This is indeed the case according to the numerical illustrations presented in the following section. In addition, based on the fact that biased progressive HMC favours final states that are more distant from the initial state compared to the standard multinomial HMC, it is also anticipated that biased progressive HMC and WALNUTS will gain slightly more from the longer trajectories produced by the flow-based measures due to fewer reversibility issues.

\section{Numerical illustrations}
\label{sec:numerical_ill}
To further illustrate the relative merits of the different variants of locally adaptive step size selection in different contexts, this section considers the application within HMC, GHMC, multinomial HMC, biased progressive HMC and WALNUTS for three challenging target distributions. In all combinations of the chosen error measure, the overarching MCMC method and the tuning parameters, 10 independent chains are computed. For the first two scenarios, each chain executes 100000 iterations during both the warmup and sampling period, leading to a total of 1 million samples across the 10 chains. For the stochastic volatility model, 10000 iterations are performed during both phases, collecting 100000 samples across the chains. Effective sample sizes (ESS) \citep{geyer1992} based on the combined chains are then computed using the \texttt{rstan::monitor} function in R package \texttt{rstan} \citep{stan2018}. 

The presented experiments below are by no means claiming that the tuning employed for each combination of the target distribution, MCMC method and error measure is the optimal in any sense. Instead, the choices with respect to the tuning parameters are made to give insight into the relative strengths and weaknesses of the different error measures in varying contexts. In particular, the actual strategies/objectives for optimally tuning the different combinations of error measure and MCMC methods may differ across this spectrum, e.g., controlling $P(c=c_{\min})$ to ensure infrequent reversibility failures for the energy error measure $\kappa_H$ as in \cite{bou2024incorporating} may be less relevant for the alternative error measures. Finding the best such strategy for each particular combination is left for future research. When tuning of $\delta$ is performed for the multinomial HMC methods, the following quantity $\Delta = \max_{j \in \mathcal C_0} H(\theta_j, \rho_j) - \min_{j \in \mathcal C_0} H(\theta_j, \rho_j)$ is relevant. In other words, $\Delta$ represents the largest deviation in the energy among the reachable states of an orbit, i.e., $    \mathcal C_0 := \left\{ k\in\{a,\ldots,b\}:r_{[k]}=1 \right\}
$, see e.g., Appendix \ref{appendix:multinomial-correctness}. Therefore, states obtained after a macro step that fails reversibility are not included in the definition of $\Delta$. The same applies for tuning of $h$, where $P(c = c_{\min})$ is the achieved rate of macro steps leading to reachable states with $c = c_{\min}$, see Appendix \ref{appendix:tuning_strategies} for further details. 

\subsection{AR(1) model} \label{sec:ar1}
The first model considered is the marginally standard Gaussian AR(1)/Toeplitz model, characterized by 
\begin{equation} \label{eq:ar1_setup}
    \theta_1 \sim N(0, 1),\;
    \theta_{t} | \theta_{t-1} \sim N( \varphi \theta_{t-1} , 1-\varphi^2), \, t = 2, \dots, d.  
\end{equation}
In particular, $d=100$ and $\varphi=0.95$ are considered here, leading to a correlated high-dimensional setting where the approach of locally adapting the step size should be beneficial. 

\subsubsection{Standard HMC and GHMC}

\begin{table}[ht]
    \centering
    \begin{tabular}{lcccc}
    \hline
         Criterion & $\delta$ & Reversibility rate & Metropolis rate &  $N_{\text{ode}}$ (millions) \\
         \toprule 
         \multicolumn{5}{c}{Standard HMC, $L = \max(1, \Tilde{L})$, $\Tilde{L} \sim \text{Poisson(16)}$, $h=0.25$}\\ 
         \hline
         H & 0.001 & 69.8\% & 69.8\% & 2913 \\
         RKNF  & 0.265 & 94.4\%& 67.2\% & 124 \\
        MRF  & 0.11 & 84.1\% & 71.7\% & 278 \\
        \hline 
        \multicolumn{5}{c}{GHMC ($L=1$), $h=1$}\\ 
        \hline 
         H & 0.0275 & 91.3\% & 90.6\% & 117 \\
         RKNF & 0.06 & 97.8\% & 90.3\% & 48 \\
        MRF  & 0.8 & 97.3\% & 89.8\% & 53 \\
        \hline
        
    \end{tabular}%
    \caption{Reversibility and diagnostic results of the AR(1) model \eqref{eq:ar1_setup} using standard HMC and GHMC. Results are based on 10 independent chains, each with 100000 iterations. Reversibility rate gives the proportion of transitions with $\prod_{j=1}^L w_j' = 1$ (see Equation \eqref{eq:metropolis_adapt_standard_hmc}) and Metropolis rate is the estimated expected value of (\ref{eq:metropolis_adapt_standard_hmc}). $N_{\text{ode}}$ is the number of gradient evaluations required to obtain the samples. Here, $c_{\min}=0$ for H and RKNF, and $c_{\min}=1$ for MRF.}
    \label{tab:ar1_hmc_reversibility}
\end{table}

\begin{table}[ht]
\centering
\begin{tabular}{cccccccccc}
\toprule
 & \multicolumn{3}{c}{Posterior Mean (Posterior SD) [Coordinate]} & \multicolumn{3}{c}{{ESS / 10000}} & \multicolumn{3}{c}{$\text{ESS} * 1000000 $/ $N_{\text{ode}}$} \\
 \cmidrule(lr){1-1} \cmidrule(lr){2-4} \cmidrule(lr){5-7} \cmidrule(lr){8-10}
 Criterion  & H & RKNF & MRF  & H & RKNF & MRF  & H & RKNF  & MRF  \\
 \cmidrule(lr){1-1} \cmidrule(lr){2-4} \cmidrule(lr){5-7} \cmidrule(lr){8-10}
Parameter &  &  &  &  &  & & & & \\
\midrule 
\multicolumn{10}{c}{Standard HMC, $L = \max(1, \Tilde{L})$, $\Tilde{L} \sim \text{Poisson(16)}$, $h=0.25$}\\ \hline \\
min ESS & 0.0 (1.0) [48] & 0.0 (1.0) [42] & 0.0 (1.0) [50] & 
14.64 & 14.24 & 15.06 & 
50 & 1150 & 542 \\
max ESS & 0.0 (1.0) [100] & 0.0 (1.0) [1] & 0.0 (1.0) [1] & 
26.99 & 25.35 & 28.45 & 
93 & 2047 & 1025 \\ \\
\hline
\multicolumn{10}{c}{GHMC ($L=1$), $h=1$}\\
\hline \\ 
min ESS & 0.0 (1.0) [64] & 0.0 (1.0) [42] & 0.0 (1.0) [63] & 
9.31 & 9.75 & 9.28 & 
794 & 2012 & 1756 \\
max ESS & 0.0 (1.0) [1] & 0.0 (1.0) [100] & 0.0 (1.0) [1] & 
17.64 & 17.75 & 17.35 & 
1504 & 3664 & 3281 \\ \\
\bottomrule
\end{tabular}%
\caption{ESS  results of AR(1) model \eqref{eq:ar1_setup} using standard HMC and GHMC. Results are based on 10 independent chains, each with 100000 iterations. Effective sample sizes are calculated across the chains. In both cases, $c_{\min}=0$ for H and RKNF, and $c_{\min}=1$ for MRF.}
\label{tab:ar1_hmc_ess}
\end{table}

Tables \ref{tab:ar1_hmc_reversibility} and \ref{tab:ar1_hmc_ess} provide diagnostic results and performance results for standard HMC and GHMC, respectively. 
For standard HMC, $h = 0.25$ is chosen for the macro step size while the number of macro steps $L$ is drawn independently in each iteration as $L=\max(1,\tilde L)$ with $\tilde L\sim\mathrm{Poisson}(16)$, in the spirit of randomized HMC \citep{bou2017randomized}, in order to avoid the resonances that a fixed integration time can produce \citep[Section~4.2]{neal2011mcmc}.
For GHMC (i.e., $L=1$), the macro step size is set to $h=1.0$ while the friction coefficient is selected so that the momentum refreshes have an autocorrelation of $\exp(-\gamma h)=0.95$. The same values of $L$ and $h$ are used for each error measure so that the proposals will have similar characteristics.

For standard HMC, tuning of the error threshold $\delta$ is performed to achieve an expected Metropolis acceptance rate  of around 70\%. For GHMC, the corresponding target value of the same quantity is 90\%. For both standard HMC and GHMC, it is seen from Table \ref{tab:ar1_hmc_reversibility} that in order to achieve the desired acceptance rate, $\delta$ has to be chosen rather small for the Hamiltonian error measure. In addition, for the Hamiltonian error measure, the reversibility rate is almost equal to the Metropolis rate, suggesting that close to all rejections are due to reversibility failures. The flow-based error measures do not experience such a scenario, allowing the selection of $\delta$ that attains the desired acceptance ratio with a large proportion of rejections being due to the inaccurate flow simulation. Since the Hamiltonian error measure requires substantially more gradient evaluations to obtain the desired acceptance ratio, the sampling performance (see Table \ref{tab:ar1_hmc_ess}) of the flow-based measures is better. In case of standard HMC, an order of magnitude is observed.

\subsubsection{Multinomial HMC, Biased progressive HMC and WALNUTS}

\begin{table}[ht]
    \centering
    \resizebox{1\columnwidth}{!}{%
\begin{tabular}{lcccccc}
        \hline
         Criterion & $h$ & $\delta$ & $P(\text{\# dead ends}\geq 1)$ &  $P(c = c_\text{min})$ & $P(\Delta \leq 0.2)$ & $N_{\text{ode}}$ (in millions) \\
         \toprule
        \multicolumn{7}{c}{Multinomial, fixed $h=0.25$, $L=32$}\\ \hline
         H & 0.25 &  0.1205 & 97.2\% & 0.1\% & 79.9\% & 231 \\
         RKNF & 0.25 & 0.0011 & 8.3\% & 0\% & 81.0\% & 1109 \\
        MRF  & 0.25 & 0.0825 & 6.0\% & 0\% & 80.0\% & 864 \\
         \hline
         \multicolumn{7}{c}{Biased progressive, fixed $h=0.25$, $m=5$}\\ \hline
         H & 0.25 & 0.12 & 96.9\% & 0.1\% & 80.6\% & 229 \\
         RKNF & 0.25 & 0.0062 & 8.3\% & 0\% & 81.4\% & 865 \\
        MRF & 0.25 & 0.0235 & 7.4\% & 0\% & 80.5\% & 1541 \\
        \hline
         \multicolumn{7}{c}{WALNUTS, fixed $h=0.25$, $m = 5$}\\ \hline
        H & 0.25 & 0.135 & 83.3\% & 0.1\% & 80.0\% & 329 \\
         RKNF & 0.25  & 0.00635 & 7.9\% & 0\% & 81.4\% & 700 \\
        MRF & 0.25 & 0.085 & 6.0\% & 0\% & 79.5\%& 663 \\
        
         \hline
         \multicolumn{7}{c}{Multinomial, tuned $h$}\\
         \hline
         H & (0.0454-0.0455) & (0.0390-0.0393) & 51.2\% & 80.0\% & 80.0\% & 225 \\
         RKNF & (0.0272-0.0273) & (0.0007-0.0007) & 3.3\% & 80.0\% & 79.8\% & 929 \\
        MRF & (0.0566-0.0568) & (0.0050-0.0050) & 7.3\% & 80.2\% & 78.8\% & 771 \\
        \hline
         \multicolumn{7}{c}{Biased progressive, tuned $h$}\\
         \hline

        H & (0.0440-0.0441) & (0.0355-0.0357) & 61.0\% & 80.0\%& 79.8\% & 305 \\
         RKNF & (0.0243-0.0243) & (0.0005-0.0005) & 4.1\% & 80.0\% & 79.8\% & 1614 \\
        MRF& (0.0497-0.0499) & (0.0033-0.0034) & 8.3\% & 80.2\% & 78.7\% & 1383 \\
        \hline
         \multicolumn{7}{c}{WALNUTS, tuned $h$}\\
         \hline
         H & (0.0472-0.0473) & (0.0432-0.0433) & 41.3\% & 80.0\% & 79.9\% & 271 \\
         RKNF & (0.0315-0.0316) & (0.0011-0.0011) & 2.7\% & 80.0\% & 79.9\% & 601 \\
        MRF  & (0.0653-0.0655) & (0.0077-0.0078) & 6.4\% & 80.1\% & 79.9\% & 535 \\
        \hline
    \end{tabular}%
    }
\caption{Diagnostic results of AR(1) model \eqref{eq:ar1_setup} for the multinomial HMC methods. Results are based on 10 independent chains, each with 100000 iterations. Effective sample sizes are calculated across the chains. Here,  $\Delta = \max_{j \in \mathcal C_0} H(\theta_j, \rho_j) - \min_{j \in \mathcal C_0} H(\theta_j, \rho_j)$ and $c_{\min}=0$ for H and RKNF, and $c_{\min}=1$ for MRF. For a given final value of $h$ after tuning, the corresponding value of $L$ for the multinomial HMC is set to $\lceil h_0 L_0 / h \rceil$, where $h_0=0.25$, $L_0=32$ and $\lceil \cdot \rceil$ denotes the ceiling function. Similarly for biased progressive HMC and WALNUTS, the corresponding value of $m$ is $\lceil \log_2((2^{m_0}-1)h_0 / h + 1) \rceil$, where $m_0 = 5$. This ensures that the (maximum) path length after tuning is at least as long as the initial path length defined by $h_0$ and either $L_0$ or $m_0$ depending on the given sampler.} 
\label{tab:diag_AR1_multinomial}
\end{table}

\begin{table}[ht]
\centering
\begin{tabular}{lccccccccc}
\toprule
 & \multicolumn{3}{c}{Posterior Mean (Posterior SD) [Coordinate]} & \multicolumn{3}{c}{{ESS / 10000}} & \multicolumn{3}{c}{$\text{ESS} * 1000000 $/ $N_{\text{ode}}$} \\
 \cmidrule(lr){1-1} \cmidrule(lr){2-4} \cmidrule(lr){5-7} \cmidrule(lr){8-10}
 Criterion  & H & RKNF  & MRF  & H & RKNF & MRF & H & RKNF  & MRF  \\
 \cmidrule(lr){1-1} \cmidrule(lr){2-4} \cmidrule(lr){5-7} \cmidrule(lr){8-10}
Parameter &  &  &  &  &  & & & & \\
\hline
\multicolumn{10}{c}{Multinomial, fixed $h=0.25$, $L=32$}\\ \hline
min ESS & 0.0 (1.0) [54] & 0.0 (1.0) [36] & 0.0 (1.0) [56] & 
2.88 & 13.65 & 14.04 & 
125 & 123 & 162 \\
max ESS & 0.0 (1.0) [100] & 0.0 (1.0) [100] & 0.0 (1.0) [1] & 
5.49 & 25.58 & 25.74 & 
238 & 231 & 298 \\
\hline
\multicolumn{10}{c}{Biased progressive, fixed $h=0.25$, $m=5$}\\ \hline
min ESS & 0.0 (1.0) [71] & 0.0 (1.0) [44] & 0.0 (1.0) [48] & 4.45 & 23.41 & 23.43 & 194 & 271 & 152 \\
max ESS & 0.0 (1.0) [100] & 0.0 (1.0) [100] & 0.0 (1.0) [100] & 8.56 & 45.00 & 45.34 & 374 & 520 & 294 \\
\hline
\multicolumn{10}{c}{WALNUTS, fixed $h=0.25$, $m =5$}\\ \hline
min ESS & 0.0 (1.0) [59] & 0.0 (1.0) [44] & 0.0 (1.0) [43] & 3.06 & 15.74 & 15.66 & 93 & 225 & 236 \\
max ESS & 0.0 (1.0) [1] & 0.0 (1.0) [1] & 0.0 (1.0) [1] & 5.68 & 27.86 & 28.16 & 173 & 398 & 425 \\
\hline
\multicolumn{10}{c}{Multinomial, tuned $h$}\\ \hline
min ESS & 0.0 (1.0) [59] & 0.0 (1.0) [43] & 0.0 (1.0) [50] & 
9.38 & 13.58 & 13.28 & 
416 & 146 & 172 \\
max ESS & 0.0 (1.0) [100] & 0.0 (1.0) [1] & 0.0 (1.0) [100] & 
17.65 & 25.22 & 24.78 & 
783 & 271 & 321 \\
\hline
\multicolumn{10}{c}{Biased progressive, tuned $h$}\\ 
\hline
min ESS & 0.0 (1.0) [47] & 0.0 (1.0) [43] & 0.0 (1.0) [53] & 
28.87 & 60.71 & 61.20 & 
946 & 376 & 442 \\
max ESS & 0.0 (1.0) [100] & 0.0 (1.0) [1] & 0.0  (1.0) [100] & 
50.69 & 98.32 & 99.47 & 
1662 & 609 & 719 \\
\hline
\multicolumn{10}{c}{WALNUTS, tuned $h$}\\ 
\hline
min ESS & 0.0 (1.0) [52] & 0.0 (1.0) [52] & 0.0 (1.0) [42] & 
14.67 & 14.48 & 15.37 & 
541 & 241 & 287 \\
max ESS & 0.0 (1.0) [100] & 0.0 (1.0) [100] & 0.0 (1.0) [1] & 
23.48 & 25.52 & 27.43 & 
866 & 425 & 513 \\
\bottomrule
\end{tabular}%
\caption{Performance results of AR(1) model \eqref{eq:ar1_setup} for the multinomial HMC methods. Results are based on 10 independent chains, each with 100000 iterations. Effective sample sizes are calculated across the chains. See the caption of Table \ref{tab:diag_AR1_multinomial} for further details.}
\label{tab:ESS_AR1_multinomial}
\end{table}

For the ``rejection-free'' methods like the multinomial HMC, biased progressive HMC and WALNUTS, two situations are considered. 
In the former, for all three samplers, $h = 0.25$. For the multinomial HMC, the number of leapfrog steps $L$ is set to 32 so that the total number of states in an orbit is $L^* = L + 1 = 33$. For the case of biased progressive HMC and WALNUTS, the (maximum for WALNUTS) number of doublings $m$ is chosen to be 5 so that the (maximum) number of states is $2^m = 2^5 = 32$ and comparable to the multinomial case. 
$\delta$ is tuned so that $\Delta \leq 0.2$ for around 80\% of the generated orbits. 
Various diagnostic information for this situation is given in the upper half of Table \ref{tab:diag_AR1_multinomial}, while information related to the sampling performance is presented in the upper half of Table \ref{tab:ESS_AR1_multinomial}. The macro step size is chosen to be rather large in this situation so that non-trivial adaptation of step sizes, that is $c>c_{\min}$, happens in the majority of the macro steps. It is observed from Table \ref{tab:diag_AR1_multinomial} that this practice leads to a large proportion of orbits having dead ends, i.e., some of the states in the generated orbit are not reachable as they are beyond the point of reversibility failures, for the energy error measure $\kappa_H$. This proportion is greatly reduced for the flow-based error measures, with MRF giving slightly fewer reversibility failures than RKNF for all three MCMC methods. 

Regarding the sampling performance, in this situation, where active use of the adaptive step size mechanism is very prevalent, it is seen from Table \ref{tab:ESS_AR1_multinomial} that the flow-based error measures recapture the higher computational cost per step (in terms of ESS/gradient evaluation) in many cases, with the effect possibly being stronger for biased progressive HMC and WALNUTS (as they benefit more from longer runs of reversible macro steps), except for the single case of MRF using biased progressive HMC. In other words, the prevalence of dead ends for $\kappa_H$ leads to much lower raw ESS values with the difference arguably being larger for the biased progressive methods.

The lower halves of Tables \ref{tab:diag_AR1_multinomial} and \ref{tab:ESS_AR1_multinomial} consider a situation where the macro step size $h$ is also tuned towards $P(c=c_{\min})=0.8$ (in addition to $\delta$ towards $P(\Delta<0.2)=0.8$ as before), which is similar to the strategy used in \cite{bou2025within}. 
Based on the tuned $h$, for the multinomial HMC, $L$ is chosen so that the path length is at least as long as the initial path length of $32 \cdot 0.25 = 8$ time units. Similarly, for biased progressive HMC and WALNUTS, $m$ is chosen so that the (maximum for WALNUTS) path length is larger than or equal to $(2 ^ 5 - 1) \cdot 0.25 = 7.75$ time units. 


Employing tuned, smaller macro step sizes reduces the prevalence of dead ends for H and RKNF (even if the orbits involve more macro steps), whereas the corresponding figures for MRF are more similar across choices of $h$ in this case. From Table \ref{tab:ESS_AR1_multinomial}, it is quite clear that the additional computational cost of the flow-based error measures does not pay off for the multinomial methods in this situation, as H produces similar values of raw ESS as the others with less computational cost. The exception is biased progressive HMC, where the number of raw ESS produced by H is roughly half of the other alternatives. Nevertheless, the energy criterion only performs only one fifth as many gradient evaluations as the flow-based criteria, leading to more ESS produced per gradient evaluation. 

\subsection{Neal's funnel model} \label{sec:funnel}
In this section, a similar study is performed for a version of Neal's funnel model \citep{neal2003slice}. More specifically, consider the following model: 
\begin{equation} \label{eq:funnel_setup}
    \theta_1 \sim N(0, 3^2),\; 
    \theta_{i}|\theta_1 \sim \text{ iid } N(0, \exp(\theta_1)), \, i = 2, \dots, d.  
\end{equation}
For this example, the dimension $d$ is set to $d = 10$. The joint distribution is known to have complicated geometric properties resembling a funnel shape such that standard HMC methods with fixed step size tend to struggle to explore the whole distribution (see e.g., \citet{kleppe2022connecting}, \citet{bou2025within} and references therein). As such, the funnel distribution typically requires substantially larger variation in the step sizes compared to the above considered AR(1) model. Fixed step size implementations would therefore require very small step sizes and large numbers of integration steps in each iteration to fully explore the target distribution.

\subsubsection{Standard HMC and GHMC}

\begin{table}[ht]
    \centering
    \begin{tabular}{lcccccc}
    \hline
         Criterion & $h$ & $\delta$ & $L$ & Reversibility rate & Metropolis rate &  $N_{\text{ode}}$ (millions) \\
         \toprule 
         \multicolumn{7}{c}{Standard HMC}\\ 
         \hline
         H & $0.5$ &  0.031 & 4 & 70.9\% & 70.2\% & 83 \\
         RKNF & $0.5$ & 0.4  & 4 & 90.9\% & 69.9\% & 23 \\
        MRF &$1.0$  & 30 & 2 & 94.0\% & 70.3\% & 27 \\
        \hline 
        \multicolumn{7}{c}{GHMC ($L=1$), $h=1$}\\ 
        \hline 
         H & 1.0 & 0.02 & 1 & 91.0\% & 90.6\% & 69 \\
         RKNF & 1.0 & 0.14 & 1 & 96.7\% & 90.6\% & 24 \\
        MRF & 1.0 & 2 & 1 & 97.2\% & 90.3\% & 27 \\
        \hline
        
    \end{tabular}%
    \caption{Diagnostic results for Neal's funnel model \eqref{eq:funnel_setup} using standard HMC and GHMC. Results are based on 10 independent chains, each with 100000 iterations. Reversibility rate gives the proportion of transitions with $\prod_{j=1}^L w_j' = 1$ (see Equation \eqref{eq:metropolis_adapt_standard_hmc}) and Metropolis rate is the estimated expected value of (\ref{eq:metropolis_adapt_standard_hmc}). $N_{\text{ode}}$ is the number of gradient evaluations required to obtain the samples.}
    \label{tab:funnel_hmc_reversibility}
\end{table}

\begin{table}[ht]
\centering
\begin{tabular}{cccccccccc}
\toprule
 & \multicolumn{3}{c}{Posterior Mean (Posterior SD) [Coordinate]} & \multicolumn{3}{c}{{ESS / 1000}} & \multicolumn{3}{c}{$\text{ESS} * 1000000 $/ $N_{\text{ode}}$} \\
 \cmidrule(lr){1-1} \cmidrule(lr){2-4} \cmidrule(lr){5-7} \cmidrule(lr){8-10}
 Criterion  & H & RKNF & MRF  & H & RKNF & MRF  & H & RKNF  & MRF  \\
 \cmidrule(lr){1-1} \cmidrule(lr){2-4} \cmidrule(lr){5-7} \cmidrule(lr){8-10}
Parameter &  &  &  &  &  & & & & \\
\midrule 
\multicolumn{10}{c}{Standard HMC }\\ \hline \\
$\theta_1$ & 0.0 (3.0) & -0.1 (2.9) & 0.0 (3.0) & 
2.67 & 2.99 & 1.10 & 
32 & 128 & 41 \\
min ESS & 0.0 (5.6) [9] & 0.0 (5.6) [3] & 1.2 (14.2) [2] & 
1.67 & 2.03 & 0.17 & 
20 & 87 & 6 \\
max ESS & -0.2 (5.0) [8] & 0.0 (5.3) [2] & 0.2 (6.3) [7] & 
3.44 & 3.30 & 1.41 & 
42 & 141 & 52 \\ \\
\hline
\multicolumn{10}{c}{GHMC ($L=1$), $h=1$}\\
\hline \\ 
$\theta_1$ & -0.1 (3.0) & 0.0 (2.9) & -0.1 (2.9) & 
1.10 & 1.19 & 1.32 & 
16 & 50 &  50  \\
min ESS & -0.1 (3.0) [1] & 0.0 (2.9) [1] & -0.1 (2.9) [1] & 
1.10 & 1.19 & 1.32 & 
16 & 50 & 50 \\
max ESS & 0.0 (5.8) [5] & 0.0 (4.7) [6] & -0.1 (4.9) [8] & 
4.62 & 9.58 & 7.57 & 
67 & 404 & 285 \\ \\
\bottomrule
\end{tabular}%
\caption{ESS results for Neal's funnel model \eqref{eq:funnel_setup} using standard HMC and GHMC. Results are based on 10 independent chains, each with 100000 iterations. Effective sample sizes are calculated across the chains.}
\label{tab:funnel_hmc_ess}
\end{table}

Tables \ref{tab:funnel_hmc_reversibility} and \ref{tab:funnel_hmc_ess} provide diagnostic information and performance results for standard HMC and GHMC ($h=1$ and $\exp(-\gamma h) = 0.95$). For standard HMC, $\delta$ is tuned towards a 70\% Metropolis acceptance probability. Note that for moderate $L$ (say $\leq 16$),  standard leapfrog HMC at macro step size $h \approx 0.40$ attains a 70\% Metropolis acceptance probability for this model (i.e., $\delta=\infty$), but produces biased samples as the narrow part of the funnel is rarely visited. Consequently, macro step sizes are chosen larger than $h=0.40$ to avoid the adaptive methods becoming insensitive to $\delta$.

For the H and RKNF criteria, $L = 4$ macro steps of size $h = 0.5$ are chosen. For H, at larger $L$ with $h=0.5$, the acceptance rate is non-monotone in $\delta$ and peaks below 70\%, so no threshold achieves the target. Reducing to $L = 4$ restores monotonicity and makes the target achievable. The required $\delta = 0.031$ incurs a reversibility failure rate of roughly 29\%. The RKNF criterion with the same combination of $h$ and $L$ is better behaved: the acceptance rate decreases monotonically as $\delta$ increases, and the target is achieved when $\delta = 0.4$ with a reversibility failure rate of only around 9\%. For the MRF criterion with $h \leq 0.5$ (and therefore a micro step size $\leq 0.25$), the acceptance rate is substantially larger than 70\% for any $\delta$. Reaching 70\% therefore requires a much larger step size, where $h = 1.0$ is chosen.


The results in Table \ref{tab:funnel_hmc_ess} are comparable with what is found for the standard HMC and GHMC in the case of the AR(1) model. Except for the min ESS parameter in the standard HMC scenario using the MRF criterion, the flow-based measures offer substantial improvements in sampling efficiency relative to H. 

\subsubsection{Multinomial HMC, Biased progressive HMC and WALNUTS}

\begin{table}[ht]
    \centering
    \resizebox{1\columnwidth}{!}{%
\begin{tabular}{lcccccc}
        \hline
         Criterion & $h$ & $\delta$ & $P(\text{\# dead ends}\geq 1)$ &  $P(c = c_\text{min})$ & $P(\Delta \leq 0.2)$ & $N_{\text{ode}}$ (in millions) \\
         \toprule
        \multicolumn{7}{c}{Multinomial, fixed $h=0.25$, $L=16$}\\ \hline
         H & 0.25 & 0.095 & 35.5\% & 75.1\% & 80.2\% & 58 \\
         RKNF & 0.25 & 0.02 & 6.1\% & 46.9\% & 80.6\% & 151 \\
        MRF & 0.25 & 0.21 & 1.9\% & 82.6\% &  81.4\% & 186  \\
         \hline
         \multicolumn{7}{c}{Biased progressive, fixed $h=0.25$, $m=4$}\\ \hline
         H & 0.25 & 0.0975 & 34.5\% & 75.0\% & 80.0\% & 56 \\
         RKNF & 0.25 & 0.02 & 5.7\% & 47.4\% & 81.9\% & 141 \\
        MRF & 0.25 & 0.23 & 1.8\% & 83.2\% & 80.9\% & 170 \\
        \hline
         \multicolumn{7}{c}{WALNUTS, fixed $h=0.25$, $m = 7$}\\ \hline
         H & 0.25 & 0.12 & 20.5\% & 93.8\% & 80.2\% & 55 \\
         RKNF & 0.25 & 0.032 & 6.7\% & 85.2\% & 80.7\% & 112\\
        MRF & 0.25 & 25 & 0.2\% & 99.7\% & 81.3\% & 136 \\
        
         \hline
         \multicolumn{7}{c}{Multinomial, tuned $h$}\\
         \hline
         H & (0.211-0.231) & (0.094-0.097) & 32.3\% & 80.0\% & 79.6\% & 57 \\
         RKNF & (0.123-0.171) & (0.018-0.024) & 2.9\% & 79.6\% & 79.0\% & 149 \\
        MRF & (0.256-0.300) & (0.204-0.234) & 2.4\% & 79.8\% & 79.5\% & 206 \\
        \hline
         \multicolumn{7}{c}{Biased progressive, tuned $h$}\\
         \hline

         H  & (0.215-0.225) & (0.080-0.082) & 40.3\% & 80.3\% & 79.4\% & 85 \\
         RKNF & (0.125-0.139) & (0.016-0.019) & 2.8\% & 79.6\% & 80.3\% & 153 \\
        MRF & (0.275-0.296) & (0.209-0.218) & 2.4\% & 79.6\% & 80.2\% & 215 \\
        \hline
         \multicolumn{7}{c}{WALNUTS, tuned $h$}\\
         \hline
         H  & (0.360-0.366) & (0.111-0.115) & 39.8\% & 80.1\% & 80.3\% & 66 \\
         RKNF & (0.268-0.273) & (0.033-0.033) & 8.2\% & 79.7\% & 79.9\% & 117 \\
        MRF & (0.535-0.544) & (0.357-0.369) & 6.7\% & 79.6\% & 80.4\% & 148 \\
        \hline
    \end{tabular}%
    }
\caption{Diagnostic results for multinomial HMC methods applied to Neal's funnel model (\ref{eq:funnel_setup}). Here $\Delta = \max_{j \in \mathcal C_0} H(\theta_j, \rho_j) - \min_{j \in \mathcal C_0} H(\theta_j, \rho_j)$. For the lower half where $h$ is tuned, the value of $L$ for the multinomial HMC is set to $\lceil h_0 L_0 / h \rceil$, where $h_0=0.25$ and $L_0=16$. Similarly for biased progressive HMC and WALNUTS, the corresponding values of $m$ are $\lceil \log_2((2^{m_0}-1)h_0 / h + 1) \rceil$, where $m_0 = 4$ for biased progressive HMC and $m_0 = 7$ for WALNUTS. This ensures that the (maximum) path length after tuning is at least as long as the initial path length defined by $h_0$ and either $L_0$ or $m_0$ depending on the given sampler.}
\label{tab:diag_funnel_multinomial}
\end{table}

\begin{table}[ht]
    \centering
    \resizebox{1\columnwidth}{!}{%
\begin{tabular}{lccccccccc}
\toprule
 & \multicolumn{3}{c}{Posterior Mean (Posterior SD) [Coordinate]} & \multicolumn{3}{c}{{ESS / 1000}} & \multicolumn{3}{c}{$\text{ESS} * 1000000 $/ $N_{\text{ode}}$} \\
 \cmidrule(lr){1-1} \cmidrule(lr){2-4} \cmidrule(lr){5-7} \cmidrule(lr){8-10}
 Criterion  & H & RKNF  & MRF  & H & RKNF & MRF & H & RKNF  & MRF  \\
 \cmidrule(lr){1-1} \cmidrule(lr){2-4} \cmidrule(lr){5-7} \cmidrule(lr){8-10}
Parameter &  &  &  &  &  & & & & \\
\hline
\multicolumn{10}{c}{Multinomial, fixed $h=0.25$, $L=16$}\\ \hline
$\theta_1$ & -0.1 (2.9) & -0.1 (2.9)  & 0.0 (3.0) & 1.89 & 2.65 & 1.47 & 32 & 18 & 8 \\
min ESS & -0.3 (7.6) [2] & -0.3 (6.2) [6] & 0.4 (6.3) [10] 
& 0.45 & 0.86 & 0.38 & 
8 & 6 & 2 \\
max ESS & 0.1 (5.2) [8] & 0.0 (4.7) [5] & 0.0 (5.3) [7] 
& 2.84 & 3.51 & 1.63 
& 49 & 23 & 9 \\
\hline
\multicolumn{10}{c}{Biased progressive, fixed $h=0.25$, $m = 4$}\\ \hline
$\theta_1$ & -0.1 (3.0) & 0.0 (3.0)  & 0.1 (3.0) & 2.15 & 1.60 & 1.45 & 39 & 11 & 9\\
min ESS & -0.1 (7.8) [4] & 0.3 (10.1) [7] & -0.5 (10.1) [9] 
& 0.55 & 0.37 & 0.41 & 
10 & 3 & 2 \\
max ESS & -0.1 (3.0) [1] & -0.1 (6.0) [6] & 0.0 (6.1) [3] 
& 2.15 & 2.05 & 2.99 
& 39 & 15 & 18 \\
\hline
\multicolumn{10}{c}{WALNUTS, fixed $h=0.25$, $m = 7$}\\ \hline
$\theta_1$ & 0.0 (3.0) & -0.1 (3.0)  & 0.0 (3.0) & 
8.38 & 9.17 & 8.11 & 
153 & 82 & 60\\
min ESS & 0.0 (3.0) [1] & 0.1 (9.0) [8] & -0.1 (9.4) [3] 
& 8.38 & 4.88 & 3.31 & 
153 & 43 & 24\\
max ESS & 0.0 (6.8) [3] & 0.0 (7.0) [5] & 0.0 (7.1) [4] 
& 38.98 & 38.82 & 31.62 
& 712 & 346 & 233 \\
\hline
\multicolumn{10}{c}{Multinomial, tuned $h$}\\ \hline
$\theta_1$ & -0.1 (2.9) & 0.1 (3.0)  & -0.2 (2.9)
& 2.47 & 1.90 & 4.99 
& 43 & 13 & 24 \\
min ESS & 0.1 (5.5) [7] & 0.2 (7.0) [7] & -0.1 (4.7) [6] 
& 1.32 & 1.00 & 2.46 
& 23 & 7 & 12 \\
max ESS & 0.0 (4.9) [6] & -0.2 (5.4) [10] & 0.0 (3.9) [10] 
& 3.07 & 2.81 & 7.97 & 
54 & 19 & 39 \\
\hline
\multicolumn{10}{c}{Biased progressive, tuned $h$}\\ 
\hline
$\theta_1$ & 0.0 (3.0) & -0.1 (2.9)   & 0.0 (3.1)
& 4.05 & 4.01 & 2.50 
& 48 & 26 & 12 \\
min ESS & 0.3 (7.9) [6] & 0.1 (5.2) [7] & -0.4 (8.0) [4] 
& 1.26 & 1.93 & 0.81 
& 15 & 13 & 4 \\
max ESS & 0.0 (3.0) [1] & -0.1 (5.0) [5] & 0.1 (6.7) [7] 
& 4.05 & 4.95 & 3.03 & 
48 & 32 & 14 \\
\hline
\multicolumn{10}{c}{WALNUTS, tuned $h$}\\ 
\hline
$\theta_1$ & 0.0 (3.0) & -0.1 (3.0)   & 0.0 (3.0)
& 8.93 & 9.59 & 5.33 
& 135 & 82 & 36 \\
min ESS &  0.0 (3.0) [1] &  -0.1 (3.0) [1] & -1.1 (22.8) [3] 
& 8.93 & 9.59 & 0.47
& 135 & 82 & 3 \\
max ESS & 0.0 (6.8) [4] & 0.0 (6.9) [3] & 0.0 (3.0) [1] 
& 44.87 & 33.57 & 5.33 & 
678 & 286 & 36 \\
\bottomrule
\end{tabular}%
}
\caption{Performance results for multinomial HMC methods applied to Neal's funnel model (\ref{eq:funnel_setup}). See caption of Table \ref{tab:diag_funnel_multinomial} for further details.}
\label{tab:ESS_funnel_multinomial}
\end{table}

Tables \ref{tab:diag_funnel_multinomial} and \ref{tab:ESS_funnel_multinomial} provide diagnostic and performance information of multinomial HMC methods for Neal's funnel model (\ref{eq:funnel_setup}). As for the AR(1) model, two situations are considered, either with fixed macro step size $h$ or with $h$ tuned towards $P(c=c_{\min})=0.8$. For the fixed step size simulations, the orbits of multinomial HMC are chosen to consist of $L^* = L + 1 = 17$ states. For the biased progressive HMC, $m = 4$ while $m = 7$ for the WALNUTS case as it can adaptively choose the final orbit size at each iteration. With these choices, a rather comparable number of gradient evaluations across the three samplers is obtained. In the fixed $h$ case with WALNUTS and MRF error measure, the chosen $h$ yields near compliance with the $P(\Delta\leq 0.2)$-criterion for a wide range of $\delta$-values. This is due to the $c_\text{min} = 1$ floor of MRF, which in this case is already adequate to obtain the desired accuracy. Thus, $\delta$ is set to 25, corresponding to the rather arbitrary largest value considered. In this setting, the adaptive step size mechanism only guards against single macro step divergences. Based on the presented results, reversibility failures for the H-measure appear not to be performance limiting for this particular model, and the increased computational cost of the flow-based error measures does not pay off in faster sampling efficiency. 


Overall, the picture of Neal's funnel distribution in the current setting is somewhat dissimilar to the one obtained from the AR(1) model. In this case, the orbits chosen by (WAL)NUTS do not involve a sufficient number of macro steps to warrant more elaborate flow-based error measures.

\subsection{Stochastic volatility model with leverage} \label{sec:sv_model}

As a more challenging real data example, a stochastic volatility model with leverage \citep[see e.g.,][]{YU2005165} is considered. Here, the same data set of $T=2515$ log-return observations $y = (y_1, \dots, y_T)$ of the S\&P500 index from October 1st 1999 to September 30th 2009 examined in \cite{Grothe26112019} is used. Let $x = (x_0, \dots, x_T)$ correspond to the  latent log-volatilities. Given the volatility scale parameter $\sigma > 0$ and the leverage parameter $\alpha \in (-1, 1)$, the joint process is modeled as 
\begin{equation*} \label{eq:joint_process_sv_model}
    (y_t, x_t)^T \mid x_{t-1} \sim N((0, x_{t-1})^T, V_{t-1}), \; V_{t-1}=\begin{pmatrix} 
    \exp(x_{t-1}) & \alpha \sigma\exp(x_{t-1}/2)\\
    \alpha \sigma\exp(x_{t-1}/2) & \sigma^2
    \end{pmatrix}. 
\end{equation*}
It is often seen for stock index data that negative returns tend to be associated with increases in volatility \citep{YU2005165}, hence negative correlation $\alpha$ is typically observed.

Instead of sampling directly from $x$, the a priori standard Gaussian driving variables $z = (z_1, \dots, z_T)$ are introduced so that the observationally equivalent model may be summarized as
\begin{align}
    \label{eq:sv_model_eq_2} z_{t} &\sim N(0,1), \, t = 1, \dots, T, \\ 
    \label{eq:sv_model_eq_3} x_t \mid z_t,x_{t-1} &= x_{t-1} + \sigma z_t, \, t = 1, \dots, T, \\ 
    \label{eq:sv_model_eq_4} y_t \mid z_t, x_{t-1} &\sim N( \alpha z_t \exp(x_{t-1} / 2), (1-\alpha^2) \exp(x_{t-1})), \, t = 1, \dots, T. 
\end{align}
To finalize the model specification, the priors 
$x_0 \sim N(\mu_0, \sigma_0^2)$, 
$\sigma^2 \sim 0.1 / \chi^2_{10}$ and $\alpha \sim \text{Unif}(-1,1)$ are specified together with the following transformations to obtain unconstrained parameters: 
\begin{equation*}
     \sigma = \exp(\upsilon) , \;
    \alpha = -1 + \frac{2 \exp(\eta)}{1 + \exp(\eta)}. 
\end{equation*}
Here, $\mu_0 \approx -8.54$ based on the logarithm of the variance of the observation vector $y_{1:T}$ and $\sigma_0$ is chosen to be 10. Therefore, the vector of sampled variables corresponds to $\theta = (\eta, \upsilon, x_0, z_{1:T})$. Funnel-effects stemming from sampling $\sigma$ jointly with $x_{0:T}$ are removed by the introduction of the driving variables $z_t$. However, $\alpha$ also interacts with both the location and scale of $z_t$ under the posterior distribution, and the model is therefore expected to exhibit funnel-like behaviour which is not easily removed through reparameterizations. Using the above parameterization, $\eta$ mixes substantially slower than the remaining components when using HMC-like methods.

For this particular example, only the scenarios of fixed $h$ GHMC and fixed $h$ and $m$ WALNUTS are considered. In both cases, the macro step size $h$ is chosen to be quite large. In other words, reduction of reversibility failures through increasing the prevalence of $c=c_{\min}$ is deliberately not attempted to achieve a situation where micro step sizes are chosen fully by the step size controller with $P(c=c_{\min})$ close to zero. As such, this situation would give more insight into the relative reversibility failure prevalence also in the case where the relevant micro step sizes are much smaller than the macro step size. In all cases, $\delta$ is tuned in order to obtain a desired acceptance probability/log-weight variability.

\subsubsection{GHMC}

\begin{table}
    \centering
    \begin{tabular}{lcccc}
    \hline
         Criterion & $\delta$ & Reversibility rate & Metropolis rate &  $N_{\text{ode}}$ (millions) \\
         \toprule 
        \multicolumn{5}{c}{GHMC ($L=1$), $h=1$}\\ 
        \hline 
         H & 0.0002 & 91.0\% & 91.0\% & 1583 \\
         RKNF & 0.06 & 95.6\% & 91.0\% & 50 \\
        MRF  & 0.1 & 92.0\% & 91.9\% & 740 \\
        \hline
        
    \end{tabular}%
    \caption{Diagnostic results for the stochastic volatility model defined by \eqref{eq:sv_model_eq_2}-\eqref{eq:sv_model_eq_4} using GHMC. Results are based on 10 independent chains, each with 10000 iterations. Reversibility rate gives the proportion of transitions with $\prod_{j=1}^L w_j' = 1$ (see Equation \eqref{eq:metropolis_adapt_standard_hmc}) and Metropolis rate is the estimated expected value of (\ref{eq:metropolis_adapt_standard_hmc}). $N_{\text{ode}}$ is the number of gradient evaluations required to obtain the samples.}
    \label{tab:sv_hmc_reversibility}
\end{table}

\begin{table}
\centering
\begin{tabular}{cccccccccc}
\toprule
 & \multicolumn{3}{c}{Posterior Mean (Posterior SD)} & \multicolumn{3}{c}{{ESS / 100}} & \multicolumn{3}{c}{$\text{ESS} * 1000000 $/ $N_{\text{ode}}$} \\
 \cmidrule(lr){1-1} \cmidrule(lr){2-4} \cmidrule(lr){5-7} \cmidrule(lr){8-10}
 Criterion  & H & RKNF & MRF  & H & RKNF & MRF  & H & RKNF  & MRF  \\
 \cmidrule(lr){1-1} \cmidrule(lr){2-4} \cmidrule(lr){5-7} \cmidrule(lr){8-10}
Parameter &  &  &  &  &  & & & & \\
\midrule 
\multicolumn{10}{c}{GHMC ($L=1$), $h=1$}\\
\hline \\ 
$\alpha$ & -0.75 (0.05) & -0.75 (0.04) & -0.75 (0.05) & 
1.88 & 2.04 & 1.70 & 
0.1 & 4.1 &  0.2  \\
$\sigma$ & 0.13 (0.01) & 0.13 (0.01) & 0.13 (0.01) & 
70.98 & 75.68 & 51.89 & 
4.5 & 152.3 & 7.0 \\
$x_0$ & -8.52 (0.33) & -8.51 (0.33) & -8.51 (0.33) & 
838.11 & 847.85 & 847.71 & 
52.9 & 1705.9 & 114.6 \\
$z_1$ & 0.02 (0.67) & 0.02 (0.67) & 0.02 (0.66) & 
758.57 & 755.62 & 766.92 & 
47.9 & 1520.3 & 103.6 \\
$z_T$ & 0.30 (0.68) & 0.30 (0.67) & 0.30 (0.66) & 
759.88 & 754.07 & 770.40 & 
48.0 & 1517.2 & 104.1 \\
\bottomrule
\end{tabular}%
\caption{ESS results for the stochastic volatility model defined by \eqref{eq:sv_model_eq_2}-\eqref{eq:sv_model_eq_4} using GHMC. Results are based on 10 independent chains, each with 10000 iterations. Effective sample sizes are calculated across the chains.}
\label{tab:sv_hmc_ess}
\end{table}

For GHMC with $h = 1$, $\delta$ is tuned to achieve a Metropolis acceptance rate close to 90\%. Ten independent chains are generated, where 10000 iterations are performed for both the warmup and the sampling period for each chain. Table \ref{tab:sv_hmc_reversibility} shows that the reversibility rate is almost identical to the obtained Metropolis rate for the Hamiltonian criterion, indicating again that almost all rejections are caused by the reversibility failures of the adaptive step size procedures. Interestingly, a similar behaviour is observed for the MRF criterion too.  The RKNF criterion appears to have rejections based on both reversibility issues and standard Metropolis rejections. Thus, in order to achieve the desired target of Metropolis acceptance rate, the RKNF criterion requires far fewer gradient evaluations. This is reflected in the performance results presented in Table \ref{tab:sv_hmc_ess}. Although the raw ESS values are rather comparable across the three approaches, the RKNF approach is indeed the most efficient when considering the number of ESS per gradient evaluation. 

\subsubsection{WALNUTS}

\begin{table}
    \centering
\begin{tabular}{lccccc}
        \hline
         Criterion & $h$ & $\delta$ & $P(\text{\# dead ends}\geq 1)$ &  $P(\Delta \leq 0.2)$ & $N_{\text{ode}}$ (in millions) \\
         \toprule
         \multicolumn{6}{c}{WALNUTS, fixed $h=0.25$, $m = 5$}\\ \hline
         H & 0.25 & 0.13 & 86.8\% &  79.0\% & 248 \\
         RKNF & 0.25 & 0.017 & 17.8\% &  80.0\% & 305\\
        MRF & 0.25 & 0.5 & 45.5\% & 79.4\% & 471 \\
        \hline
    \end{tabular}%
\caption{Diagnostic results for WALNUTS applied to the stochastic volatility model  \eqref{eq:sv_model_eq_2}-\eqref{eq:sv_model_eq_4}. Here $\Delta = \max_{j \in \mathcal C_0} H(\theta_j, \rho_j) - \min_{j \in \mathcal C_0} H(\theta_j, \rho_j)$. 
}
\label{tab:diag_sv_multinomial}
\end{table}

\begin{table}
\centering
\begin{tabular}{cccccccccc}
\toprule
 & \multicolumn{3}{c}{Posterior Mean (Posterior SD)} & \multicolumn{3}{c}{{ESS / 100}} & \multicolumn{3}{c}{$\text{ESS} * 1000000 $/ $N_{\text{ode}}$} \\
 \cmidrule(lr){1-1} \cmidrule(lr){2-4} \cmidrule(lr){5-7} \cmidrule(lr){8-10}
 Criterion  & H & RKNF & MRF  & H & RKNF & MRF  & H & RKNF  & MRF  \\
 \cmidrule(lr){1-1} \cmidrule(lr){2-4} \cmidrule(lr){5-7} \cmidrule(lr){8-10}
Parameter &  &  &  &  &  & & & & \\
\midrule 
\multicolumn{10}{c}{WALNUTS, fixed $h = 0.25$, $m = 5$}\\
\hline \\ 
$\alpha$ & -0.75 (0.05) & -0.75 (0.05) & -0.75 (0.04) & 
11.72 & 15.11 & 14.34 & 
4.7 & 5.0 &  3.0  \\
$\sigma$ & 0.13 (0.01) & 0.13 (0.01) & 0.13 (0.01) & 
240.83 & 379.31 & 341.30 & 
97.0 & 124.5 & 72.4 \\
$x_0$& -8.52 (0.33) & -8.51 (0.33) & -8.51 (0.33) & 
508.94 & 1367.02 & 989.40 & 
205.0 & 448.8 & 210.0 \\
$z_1$& 0.02 (0.66) & 0.02 (0.66) & 0.02 (0.66) & 
512.18 & 1986.52 & 1212.67 & 
206.3 & 652.1 & 257.4 \\
$z_T$ & 0.30 (0.67) & 0.30 (0.66) & 0.31 (0.66) & 
504.37 & 1855.95 & 1174.35 & 
203.2 & 609.3 & 249.2 \\
\bottomrule
\end{tabular}%
\caption{Performance results for WALNUTS applied to the stochastic volatility model \eqref{eq:sv_model_eq_2}-\eqref{eq:sv_model_eq_4}. See caption of Table \ref{tab:diag_sv_multinomial} for further details.}
\label{tab:ESS_sv_multinomial}
\end{table}

Next, the stochastic volatility model is examined using WALNUTS with a fixed macro step size of $h = 0.25$ and $m = 5$, and $\delta$ is tuned so that $P(\Delta \leq 0.2)$ is roughly 0.8 similar to the previous experiments. Table \ref{tab:diag_sv_multinomial} provides diagnostic statistics, with the RKNF approach leading to the smallest ratio of orbits with at least one dead end due to reversibility issues. The MRF criterion has again a higher ratio of orbits with at least one dead end than RKNF, but far lower than for the Hamiltonian criterion, for which roughly 87\% of the 100000 orbits will experience a reversibility issue in one of the directions. 

Consequently, based on the results of Table \ref{tab:ESS_sv_multinomial}, the RKNF criterion leads to the largest values of raw ESS and ESS per gradient evaluation. However, although the raw ESS of the MRF criterion is larger, it uses roughly twice as many gradient evaluations as the Hamiltonian criterion. 
Therefore, the reduction in the reversibility failure ratio only leads to a slight increase in the number of ESS per gradient evaluation for the latent variables. 

To conclude, for this particular, more realistic model, and without the aid of a carefully tuned macro step size, it clearly appears that RKNF provides substantial computational advantage, even for multinomial methods such as WALNUTS, especially for the latent states.

\section{Discussion}\label{sec:discussion}

This paper has carefully explained and illustrated the locally adaptive step size procedure of \cite{bou2025within} under more general assumptions. Also, further error measures with the aim of reducing reversibility failures have been developed and introduced. It is also shown how to introduce the locally adaptive step sizes into standard HMC and Metropolis-corrected kinetic Langevin/GHMC. Finally, a numerical study of the relative merits of the proposed error measures in the context of different Hamiltonian flow-based MCMC methods has been presented.



It is believed that any general purpose Hamiltonian flow-based MCMC code should employ such adaptive step size procedures. In particular, the most significant finding from the numerical study is that non-multinomial HMC methods such as the standard HMC and kinetic Langevin/GHMC may benefit more from improved error measures. In contrast, energy errors seem to be working well for certain models using the conventional multinomial HMC. Still, for the two models considered in Sections \ref{sec:ar1} and \ref{sec:funnel}, the orbit sizes are relatively modest, and it is conceivable that other models which would require larger orbits could benefit more from the flow-based error measures. For instance, the results from the stochastic volatility model examined in Section \ref{sec:sv_model} already show that one of the flow-based error measures outperforms the energy error when using WALNUTS. 

The paper should to some degree be read as an invitation to further work, where optimal tuning for given combinations of error measures and overarching MCMC methods should be considered. In addition, extending the locally adaptive step sizes also to other forms of flow, as indicated in Section \ref{sec:further_remarks}, and the BAOAB integrator for Langevin-type models are currently being investigated.

\bibliographystyle{plainnat}
\bibliography{main}

\newpage

\appendix
\begin{center}
\Large Supplementary material to "Locally adaptive step sizes for Hamiltonian Monte Carlo methods" by Jimmy Huy Tran, Tore Selland Kleppe and Nawaf Bou-Rabee
\end{center}

\section{Proofs} \label{appendix:proofs}
\subsection{Proof of Proposition \ref{prop_flow_criterion}} \label{appendix:proof_prop_flow_criterion}
\begin{proof}
Recall that $\bar\Phi_{h/l}$ denotes the embedded third-order RKN
comparison map defined by
\eqref{eq:third_order_integrator_flow_nr1}-\eqref{eq:third_order_integrator_flow_nr3}. This map is used only to
estimate the error of a leapfrog microstep; the numerical trajectory
itself is advanced using the leapfrog map $\Phi_{h/l}$.

Let
\[
z_0=(\theta,\rho),
\qquad
z_k:=\Phi_{h/l}^k(z_0),
\qquad
k=0,\ldots,l.
\]
For each $k=1,\ldots,l$, define
\begin{align*}
    f_k
    &:=
    \left\|
        z_k-\bar\Phi_{h/l}(z_{k-1})
    \right\|,\\
    b_k
    &:=
    \left\|
        \mathcal F(z_{k-1})
        -
        \bar\Phi_{h/l}\circ\mathcal F(z_k)
    \right\|.
\end{align*}
Thus,
\[
\kappa_{\mathrm{RKNF}}(z_0,h,l)
=
\max_{1\leq k\leq l}\max\{f_k,b_k\}.
\]

Now consider the leapfrog trajectory beginning at the
momentum-flipped endpoint of the original trajectory:
\[
\widetilde z_0:=\mathcal F(z_l),
\qquad
\widetilde z_j:=\Phi_{h/l}^j(\widetilde z_0),
\qquad
j=0,\ldots,l.
\]
The time reversibility of the leapfrog integrator implies that
\[
\Phi_{h/l}^j\circ\mathcal F
=
\mathcal F\circ\Phi_{h/l}^{-j}.
\]
Consequently,
\begin{align*}
    \widetilde z_j
    &=
    \Phi_{h/l}^j\circ\mathcal F\circ
    \Phi_{h/l}^l(z_0)\\
    &=
    \mathcal F\circ\Phi_{h/l}^{-j}\circ
    \Phi_{h/l}^l(z_0)\\
    &=
    \mathcal F\circ\Phi_{h/l}^{l-j}(z_0)\\
    &=
    \mathcal F(z_{l-j}).
\end{align*}

For the trajectory $\{\widetilde z_j\}_{j=0}^l$, define the
corresponding error terms, for $j=1,\ldots,l$, by
\begin{align*}
    \widetilde f_j
    &:=
    \left\|
        \widetilde z_j
        -
        \bar\Phi_{h/l}(\widetilde z_{j-1})
    \right\|,\\
    \widetilde b_j
    &:=
    \left\|
        \mathcal F(\widetilde z_{j-1})
        -
        \bar\Phi_{h/l}\circ\mathcal F(\widetilde z_j)
    \right\|.
\end{align*}
Using $\widetilde z_j=\mathcal F(z_{l-j})$, one obtains
\begin{align*}
    \widetilde f_j
    &=
    \left\|
        \mathcal F(z_{l-j})
        -
        \bar\Phi_{h/l}\circ\mathcal F(z_{l-j+1})
    \right\|\\
    &=
    b_{l-j+1},
\end{align*}
and, using $\mathcal F\circ\mathcal F=\operatorname{Id}$,
\begin{align*}
    \widetilde b_j
    &=
    \left\|
        \mathcal F\circ\mathcal F(z_{l-j+1})
        -
        \bar\Phi_{h/l}\circ\mathcal F\circ
        \mathcal F(z_{l-j})
    \right\|\\
    &=
    \left\|
        z_{l-j+1}
        -
        \bar\Phi_{h/l}(z_{l-j})
    \right\|\\
    &=
    f_{l-j+1}.
\end{align*}

Since $j\mapsto l-j+1$ is a bijection of
$\{1,\ldots,l\}$, it follows that
\begin{align*}
    \kappa_{\mathrm{RKNF}}(\widetilde z_0,h,l)
    &=
    \max_{1\leq j\leq l}
    \max\{\widetilde f_j,\widetilde b_j\}\\
    &=
    \max_{1\leq j\leq l}
    \max\{b_{l-j+1},f_{l-j+1}\}\\
    &=
    \max_{1\leq k\leq l}\max\{f_k,b_k\}\\
    &=
    \kappa_{\mathrm{RKNF}}(z_0,h,l).
\end{align*}
Therefore, $\kappa_{\mathrm{RKNF}}$ is reversible in the sense of
\eqref{eq:error_measure_reversible_def}.
\end{proof}

\subsection{Proof of Proposition \ref{prop_bs_criterion}} \label{appendix:proof_prop_bs_criterion}

\begin{proof}
Let $z=(\theta,\rho)$, and define the fine forward endpoint,
the coarse forward endpoint, and the coarse endpoint computed from
the momentum-flipped fine endpoint, respectively, by
\begin{align*}
    z_f
    &:=
    \Phi_{h/l}^{l}(z),\\
    z_c
    &:=
    \Phi_{2h/l}^{l/2}(z),\\
    z_b
    &:=
    \Phi_{2h/l}^{l/2}\circ\mathcal F(z_f).
\end{align*}
With this notation, the MRF error measure is
\begin{equation*}
\label{eq:mrf-proof-forward}
    \kappa_{\mathrm{MRF}}(z,h,l)
    =
    \max\left\{
        \left\|z_f-z_c\right\|,
        \left\|\mathcal F(z)-z_b\right\|
    \right\}.
\end{equation*}

To verify reversibility, consider the momentum-flipped fine endpoint
\[
\widetilde z:=\mathcal F(z_f).
\]
Define the corresponding fine and coarse endpoints starting from
$\widetilde z$ by
\begin{align*}
    \widetilde z_f
    &:=
    \Phi_{h/l}^{l}(\widetilde z),\\
    \widetilde z_c
    &:=
    \Phi_{2h/l}^{l/2}(\widetilde z),\\
    \widetilde z_b
    &:=
    \Phi_{2h/l}^{l/2}\circ
    \mathcal F(\widetilde z_f).
\end{align*}

By the time reversibility of the leapfrog integrator,
\[
\Phi_{h/l}^{l}\circ\mathcal F\circ
\Phi_{h/l}^{l}
=
\mathcal F.
\]
Therefore,
\[
\widetilde z_f
=
\Phi_{h/l}^{l}\circ\mathcal F(z_f)
=
\Phi_{h/l}^{l}\circ\mathcal F\circ
\Phi_{h/l}^{l}(z)
=
\mathcal F(z).
\]
Moreover, by the definitions of $z_b$ and $\widetilde z$,
\[
\widetilde z_c
=
\Phi_{2h/l}^{l/2}(\widetilde z)
=
\Phi_{2h/l}^{l/2}\circ\mathcal F(z_f)
=
z_b.
\]
Finally, since
$\widetilde z_f=\mathcal F(z)$,
\[
\widetilde z_b
=
\Phi_{2h/l}^{l/2}\circ
\mathcal F(\widetilde z_f)
=
\Phi_{2h/l}^{l/2}\circ
\mathcal F\circ\mathcal F(z)
=
\Phi_{2h/l}^{l/2}(z)
=
z_c.
\]

Applying the definition of the MRF error measure at
$\widetilde z$ gives
\begin{align*}
    \kappa_{\mathrm{MRF}}(\widetilde z,h,l)
    &=
    \max\left\{
        \left\|\widetilde z_f-\widetilde z_c\right\|,
        \left\|\mathcal F(\widetilde z)-\widetilde z_b\right\|
    \right\}\\
    &=
    \max\left\{
        \left\|\mathcal F(z)-z_b\right\|,
        \left\|z_f-z_c\right\|
    \right\}\\
    &=
    \kappa_{\mathrm{MRF}}(z,h,l).
\end{align*}
Thus, $\kappa_{\mathrm{MRF}}$ is reversible in the sense of
\eqref{eq:error_measure_reversible_def}.
\end{proof}

\section{Order verification for the embedded RKN map}
\label{appendix:rkn_order}

\begin{lemma}\label{lem:rkn_order}
Assume that $\log\widetilde{\pi}$ is four times continuously
differentiable in a neighborhood of the exact Hamiltonian trajectory.
Then, the map $\bar\Phi_\epsilon$ defined by
\eqref{eq:third_order_integrator_flow_nr1}--
\eqref{eq:third_order_integrator_flow_nr3} satisfies
\[
    \bar\Phi_\epsilon(\theta,\rho)
    =
    \varphi_\epsilon(\theta,\rho)
    +
    O(\epsilon^4)
    \qquad\text{as }\epsilon\to0.
\]
Thus, $\bar\Phi_\epsilon$ is a third-order one-step method. Moreover,
because leapfrog is a second-order method,
\[
    \Phi_\epsilon(\theta,\rho)
    -
    \bar\Phi_\epsilon(\theta,\rho)
    =
    O(\epsilon^3),
\]
so their difference provides an estimate of the local leapfrog error.
\end{lemma}

\begin{proof}
For the proof, write
\[
    v=M^{-1}\rho,
    \qquad
    a(\theta)=M^{-1}\nabla_\theta\log\widetilde{\pi}(\theta),
\]
and let
\[
    A=Da(\theta),
    \qquad
    B=D^2a(\theta).
\]
Taylor expansion of the exact Hamiltonian flow gives
\begin{align*}
    \theta(\epsilon)
    &=
    \theta
    +\epsilon v
    +\frac{\epsilon^2}{2}a(\theta)
    +\frac{\epsilon^3}{6}Av
    +O(\epsilon^4),\\
    v(\epsilon)
    &=
    v
    +\epsilon a(\theta)
    +\frac{\epsilon^2}{2}Av
    +\frac{\epsilon^3}{6}
        \left(
            Aa(\theta)+B[v,v]
        \right)
    +O(\epsilon^4),
\end{align*}
where $B[\cdot,\cdot]$ denotes the vector-valued bilinear form induced by the third derivative tensor $B$.

Let $(\theta',\rho')=\Phi_\epsilon(\theta,\rho)$ and set
$v'=M^{-1}\rho'$. The leapfrog position update gives
\[
    \theta'
    =
    \theta+\epsilon v+\frac{\epsilon^2}{2}a(\theta).
\]
Consequently,
\[
    a(\theta')
    =
    a(\theta)
    +\epsilon Av
    +\frac{\epsilon^2}{2}
        \left(
            Aa(\theta)+B[v,v]
        \right)
    +O(\epsilon^3).
\]
The leapfrog momentum update therefore implies
\[
    v'
    =
    v+\epsilon a(\theta)
    +\frac{\epsilon^2}{2}Av
    +\frac{\epsilon^3}{4}
        \left(
            Aa(\theta)+B[v,v]
        \right)
    +O(\epsilon^4).
\]

Using the definition of $\theta_m$, one obtains
\[
    \theta_m
    =
    \theta
    +\frac{\epsilon}{2}v
    +\frac{\epsilon^2}{8}a(\theta)
    +O(\epsilon^3),
\]
and hence
\[
    a(\theta_m)
    =
    a(\theta)
    +\frac{\epsilon}{2}Av
    +\frac{\epsilon^2}{8}
        \left(
            Aa(\theta)+B[v,v]
        \right)
    +O(\epsilon^3).
\]
Substituting into
\eqref{eq:third_order_integrator_flow_nr2} gives
\[
    \theta''
    =
    \theta
    +\epsilon v
    +\frac{\epsilon^2}{2}a(\theta)
    +\frac{\epsilon^3}{6}Av
    +O(\epsilon^4),
\]
while substituting into
\eqref{eq:third_order_integrator_flow_nr3} gives
\[
    M^{-1}\rho''
    =
    v
    +\epsilon a(\theta)
    +\frac{\epsilon^2}{2}Av
    +\frac{\epsilon^3}{6}
        \left(
            Aa(\theta)+B[v,v]
        \right)
    +O(\epsilon^4).
\]
These expansions agree with the exact-flow expansions through terms
of order $\epsilon^3$, proving the result.
\end{proof}

\section{Correctness of locally adaptive multinomial HMC}
\label{appendix:multinomial-correctness}

Throughout this appendix, $\mathcal L$ is the deterministic rule \eqref{eq:gen_I_def} (or its truncated version), $H\circ\mathcal F=H$, and the momentum is refreshed from $N(0,M)$ at the beginning of every iteration and discarded at the end.

\paragraph{Orbits.}
Given a current state $z_{[0]}=(\theta_{[0]},\rho_{[0]})$ and a window $(a,b)$ with $a\leq 0\leq b$, the orbit $\mathcal O=\{z_{[k]}\}_{k=a}^{b}$ is defined by
\begin{align*}
  z_{[k]} &= \Phi^{l_{k-1,k}}_{h/l_{k-1,k}}(z_{[k-1]}),\quad
  l_{k-1,k}=\mathcal L(z_{[k-1]},h,\delta), \qquad k=1,\dots,b,\\
  z_{[k]} &= \mathcal F\circ\Phi^{l_{k,k+1}}_{h/l_{k,k+1}}\circ\mathcal F(z_{[k+1]}),\quad
  l_{k,k+1}=\mathcal L(\mathcal F(z_{[k+1]}),h,\delta), \qquad k=-1,\dots,a,
\end{align*}
with reversibility indicators $w'_{[k-1,k]}=\mathbf 1\{\mathcal L(z_{[k-1]},h,\delta)=\mathcal L(\mathcal F(z_{[k]}),h,\delta)\}$, $k=a+1,\dots,b$, reversibility weights $r_{[k]}$ as in \eqref{eq:reversibility_weights}, and
\[
    \mathcal C_0 := \left\{ k\in\{a,\ldots,b\}:r_{[k]}=1 \right\}.
\]
Since the indicators are attached to consecutive pairs, $\mathcal C_0$ is an interval of integers containing $0$; it is the set of indices that can be reached from index zero without crossing a macro step that fails its reversibility check. For conventional multinomial HMC, the window is $a=b-L$ with $b$ uniformly distributed on $\{0,\dots,L\}$, and the next state is $z_{[J]}$, where
\begin{equation}
    \mathbb P(J=k\mid\mathcal O)
    =
    \frac{r_{[k]}\exp\bigl(-H(z_{[k]})\bigr)}{S(\mathcal O)},
    \qquad
    S(\mathcal O):=\sum_{n=a}^{b} r_{[n]}\exp\bigl(-H(z_{[n]})\bigr).
    \label{eq:multinomial-correctness-weights}
\end{equation}

\begin{lemma}[Re-indexing] \label{lem:reindexing}
Let $\mathcal O=\{z_{[k]}\}_{k=a}^{b}$ be an orbit with origin $z_{[0]}$ and let $j\in\mathcal C_0$. Let $\mathcal O^*=\{z^*_{[s]}\}_{s=a^*}^{b^*}$ be the orbit with origin $z^*_{[0]}:=\mathcal F(z_{[j]})$ and window $(a^*,b^*)=(j-b,\,j-a)$. Then
\begin{equation*}
  z^*_{[s]}=\mathcal F(z_{[j-s]}),\qquad
  w'^{*}_{[s-1,s]}=w'_{[j-s,\,j-s+1]},\qquad
  r^*_{[s]}=r_{[j-s]},\qquad s=a^*,\dots,b^*,
  \label{eq:reindexing_identities}
\end{equation*}
and in particular $\mathcal C_0^*=\{j-k:k\in\mathcal C_0\}$ and $S(\mathcal O^*)=S(\mathcal O)$.
\end{lemma}

\begin{proof}
Assume $j\geq 1$; the case $j\leq -1$ is symmetric and the case $j=0$ is trivial. Since $j\in\mathcal C_0$, $w'_{[k-1,k]}=1$ for $k=1,\dots,j$, that is,
\begin{equation}
  \mathcal L(\mathcal F(z_{[k]}),h,\delta)=l_{k-1,k},\qquad k=1,\dots,j.
  \label{eq:reindexing_counts}
\end{equation}
Recall also that time reversibility of the leapfrog integrator gives $\Phi^{l}_{h/l}\circ\mathcal F=\mathcal F\circ(\Phi^{l}_{h/l})^{-1}$ for every $l$.

\emph{Indices $0\leq s\leq j$.} $z^*_{[s]}=\mathcal F(z_{[j-s]})$ can be shown by induction on $s$. For $s=0$, this is the definition of $\mathcal O^*$. If it holds for $s-1$, then the forward count from $z^*_{[s-1]}=\mathcal F(z_{[j-s+1]})$ is $\mathcal L(\mathcal F(z_{[j-s+1]}),h,\delta)=l_{j-s,j-s+1}$ by \eqref{eq:reindexing_counts}, and with $l=l_{j-s,j-s+1}$,
\[
  z^*_{[s]}=\Phi^{l}_{h/l}(\mathcal F(z_{[j-s+1]}))=\mathcal F\big((\Phi^{l}_{h/l})^{-1}(z_{[j-s+1]})\big)=\mathcal F(z_{[j-s]}).
\]

\emph{Indices $s<0$} (original indices $k=j-s>j$). Backward steps of $\mathcal O^*$ from $z^*_{[s+1]}=\mathcal F(z_{[j-s-1]})$ use the count $\mathcal L(\mathcal F(z^*_{[s+1]}),h,\delta)=\mathcal L(z_{[j-s-1]},h,\delta)=l_{j-s-1,j-s}$, which is the original forward count, and produce
$\mathcal F\circ\Phi^{l}_{h/l}\circ\mathcal F(\mathcal F(z_{[j-s-1]}))=\mathcal F\big(\Phi^{l}_{h/l}(z_{[j-s-1]})\big)=\mathcal F(z_{[j-s]})$.

\emph{Indices $s>j$} (original indices $k=j-s<0$). Forward steps of $\mathcal O^*$ from $z^*_{[s-1]}=\mathcal F(z_{[j-s+1]})$ use the count $\mathcal L(\mathcal F(z_{[j-s+1]}),h,\delta)=l_{j-s,j-s+1}$, which is the original backward count, and produce
$\Phi^{l}_{h/l}(\mathcal F(z_{[j-s+1]}))=\mathcal F\big((\Phi^{l}_{h/l})^{-1}(z_{[j-s+1]})\big)=\mathcal F(z_{[j-s]})$.

Thus, only the macro steps between indices $0$ and $j$ are traversed in the opposite direction in $\mathcal O^*$; all other macro steps are traversed in the same direction as in $\mathcal O$, with the same micro step counts. For the indicators,
\[
  w'^{*}_{[s-1,s]}
  =\mathbf 1\{\mathcal L(z^*_{[s-1]},h,\delta)=\mathcal L(\mathcal F(z^*_{[s]}),h,\delta)\}
  =\mathbf 1\{\mathcal L(\mathcal F(z_{[j-s+1]}),h,\delta)=\mathcal L(z_{[j-s]},h,\delta)\}
  =w'_{[j-s,\,j-s+1]},
\]
where the last equality uses only that the indicator is a symmetric function of the pair of states. Hence, $r^*_{[s]}$ is the product of the original indicators over the macro steps between indices $j-s$ and $j$, while $r_{[j-s]}$ is the product over the macro steps between $j-s$ and $0$. 
Now, when $s > j$, $r_{[j-s]} = \prod_{n=j-s+1}^0 w'_{[n-1,n]}$ as $k = j - s < 0$. On the other hand,
\begin{align*}
    r^*_{[s]} &= \prod_{n^*=1}^s w'^*_{[n^*-1, n^*]} = \prod_{n^*=1}^s w'_{[j-n^*, j-n^*+1]} = \prod_{n = j-s+1}^{j} w'_{[n - 1, n]} \\
    &=\prod_{n=1}^j w'_{[n-1, n]} \prod_{n=j-s+1}^0 w'_{[n - 1, n]} \\
    &= (\prod_{n=1}^j w'_{[n-1, n]}) r_{[j-s]},
\end{align*}
where the order of the product in the first line from $n^*$ to $n$ is due to $s > j$ and $j > 0$. 
Because all indicators between $0$ and $j$ are equal to one, one arrives at $r_{[j-s]} = r^*_{[s]}$ as the two differ by the product of the weights between 0 and $j$. If $0 \leq s \leq j$, both contain weights between $0$ and $j$, so $r_{[j-s]} = r^*[s] = 1$. In a similar manner, when $s < 0$, $r_{[j - s]} = \prod_{n=1}^{j-s} w'_{[n-1, n]} = \prod_{n=1}^{j} w'_{[n-1, n]} \prod_{n=s+1}^{0} w'_{[j-n, j-n+1]}$ since $s < 0$ corresponds to $k = j -s > j$. However, $r^*[s] = \prod_{n^*=s+1}^0 w'^*_{[n^*-1,n^*]} = \prod_{n^*=s+1}^0 w'_{[j-n^*,j-n^*+1]}$. Therefore, $r_{[j-s]} = r^*_{[s]}$ as they differ again by the products of weights between 0 and $j$. In any case, both reduce either to 1 or the product over the macro steps between $j-s$ and the nearest endpoint of the interval $[0, j]$, so $r_{[j-s]} = r^*_{[s]}$ overall. 
Finally, the statements about $\mathcal C_0^*$ and $S$ follow, the latter using $H\circ\mathcal F=H$.
\end{proof}

\begin{prop} \label{prop:multinomial_correct}
The locally adaptive multinomial HMC transition, defined by refreshing $\rho\sim N(0,M)$, drawing $b$ uniformly on $\{0,\dots,L\}$, building $\mathcal O$ with $a=b-L$, drawing $J$ from \eqref{eq:multinomial-correctness-weights} and returning $\theta_{[J]}$, is a $\pi$-reversible Markov kernel on $\mathbb R^d$.
\end{prop}

\begin{proof}
Fix $b\in\{0,\dots,L\}$ and $j\in\{b-L,\dots,b\}$. For an initial state $z\in\mathbb R^{2d}$, write $\mathcal O_b(z)$ for the orbit with origin $z$ and window $(b-L,b)$, and let $E_{b,j}:=\{z: r_{[j]}=1 \text{ in } \mathcal O_b(z)\}$. Define $\Psi_{b,j}:E_{b,j}\to\mathbb R^{2d}$ by $\Psi_{b,j}(z):=\mathcal F(z_{[j]})$. Let $b^*:=L+j-b$; then $b\mapsto b^*$ is a bijection of $\{b: b-L\leq j\leq b\}$ onto itself, and $(b^*-L,b^*)=(j-b,\,j-(b-L))$ is exactly the window of Lemma~\ref{lem:reindexing}. By that lemma, $\Psi_{b,j}(E_{b,j})\subset E_{b^*,j}$ and $\Psi_{b^*,j}\circ\Psi_{b,j}=\mathrm{id}$ on $E_{b,j}$, because $\mathcal F(z^*_{[j]})=\mathcal F(\mathcal F(z_{[0]}))=z$. Hence $\Psi_{b,j}$ is a bijection of $E_{b,j}$ onto $E_{b^*,j}$ with inverse $\Psi_{b^*,j}$. Moreover, $E_{b,j}$ is the countable disjoint union of the measurable sets on which the vector of micro step counts of $\mathcal O_b(z)$ is constant. Since $\mathcal L$ is measurable, and on each such set, $\Psi_{b,j}$ is the restriction of a fixed composition of leapfrog maps and $\mathcal F$, hence of a volume preserving map, $\Psi_{b,j}$ preserves the Lebesgue measure.

Let $P$ be the phase-space kernel that maps $z$ to $\mathcal F(z_{[J]})$; the momentum flip does not affect the position and is included so that the maps $\Psi_{b,j}$ are involutive in the above sense. Explicitly,
\[
  P(z,B)=\sum_{b=0}^{L}\frac{1}{L+1}\sum_{j=b-L}^{b}
  \mathbf 1_{E_{b,j}}(z)\,\frac{\exp(-H(z_{[j]}))}{S(\mathcal O_b(z))}\,
  \mathbf 1_B\bigl(\Psi_{b,j}(z)\bigr).
\]
For measurable $A,B\subset\mathbb R^{2d}$,
\[
  \int_A e^{-H(z)}P(z,B)\,dz
  =\sum_{b}\sum_{j}\frac{1}{L+1}\int_{A\cap E_{b,j}\cap\Psi_{b,j}^{-1}(B)}
  \frac{e^{-H(z)}e^{-H(z_{[j]})}}{S(\mathcal O_b(z))}\,dz .
\]
First, substitute $z'=\Psi_{b,j}(z)$ in the $(b,j)$ term. The Jacobian is one; $e^{-H(z)}e^{-H(z_{[j]})}=e^{-H(z')}e^{-H(z'_{[j]})}$ because $H(z')=H(\mathcal F(z_{[j]}))=H(z_{[j]})$ and $z'_{[j]}=\mathcal F(z)$ by Lemma~\ref{lem:reindexing}; $S(\mathcal O_b(z))=S(\mathcal O_{b^*}(z'))$ by the same lemma; and the domain of integration becomes $B\cap E_{b^*,j}\cap\Psi_{b^*,j}^{-1}(A)$. Summing over $(b,j)$ is the same as summing over $(b^*,j)$. Hence
\[
  \int_A e^{-H(z)}P(z,B)\,dz=\int_B e^{-H(z)}P(z,A)\,dz,
\]
that is, $P$ satisfies detailed balance with respect to $\mu\propto e^{-H}$ on $\mathbb R^{2d}$.

Finally, the position kernel is $K(\theta,d\theta')=\int N(d\rho;0,M)\,P\bigl((\theta,\rho),d\theta'\times\mathbb R^d\bigr)$. Since $\mu(d\theta\,d\rho)=\pi(d\theta)N(d\rho;0,M)$, integrating the phase-space detailed balance identity over $\rho$ and $\rho'$ gives $\pi(d\theta)K(\theta,d\theta')=\pi(d\theta')K(\theta',d\theta)$.
\end{proof}

\paragraph{Remark.}
Dead ends require no special treatment: an index $k\notin\mathcal C_0$ has $r_{[k]}=0$, is never selected, and does not contribute to $S(\mathcal O)$. The implementation may therefore stop building the orbit at the first failing macro step in each direction.

\section{Comments regarding Table \ref{tab:criteria_computational_cost}} \label{appendix:criteria_computational_cost}

For the Hamiltonian criterion, if $c^*$ is indeed the smallest value such that the criterion is satisfied for a given threshold $\delta$, the criterion must not be satisfied for $c = 0, \dots, c^* - 1$. This implies that for each value of $c$, $l = 2^c$ gradient evaluations are done in order to perform $l$ leapfrog micro steps. Thus, the forward pass will evaluate the gradient a total of $\sum_{c = 0}^{c^*} 2^c = 2^{c^*+1}-1$ times. In a similar manner, because of Proposition \ref{prop:approx_reverse}, to satisfy the reversibility condition, the backward pass must perform $\sum_{c = 0}^{c^* - 1} 2^c = 2^{c^*}-1$ gradient evaluations. 

For the MRF criterion, if $c^* = 1$, then it is easy to see that the required number of gradient evaluations for the forward pass is 4. However, for $c^* > 1$, the finer state from $c - 1$ can then be reused as the coarser state when checking the criterion for a given value of $c = 2, \dots, c^*$. Therefore, based on Figure \ref{fig:bs_illustration}, which corresponds to the case of $c = 2$, the upper path is the finer state from $c = 1$ and has therefore been computed. Thus, only $4+2=6$ evaluations are needed and the total number of gradient evaluations for the forward pass when $c^* = 2$ is $4+6=10$. More generally, for a given $c^*$, the number of gradient evaluations for the forward pass becomes $3(2^{c^*}) - 2$. 

Regarding the backward pass, if $c^* = 1$, it is clear from Proposition \ref{prop:approx_reverse} (using the reversibility of $\kappa_{MRF}$ stated in Proposition \ref{prop_bs_criterion}) that one does not need to compute the backward pass as the step is naturally reversible, i.e., the number of gradient evaluations for the backward pass is zero in this case. For $c^* = 2$, Proposition \ref{prop:approx_reverse} states that it is only necessary to compute the backward pass for $c = 1$. In that case, one can observe that the finer forward state during the backward pass is indeed the backward state from the forward pass for $c = c^* = 2$. Consequently, reusing this computed path, one only needs to compute the coarser forward state and the backward state of this backward pass, where each one corresponds to a single gradient evaluation. The total number of gradient evaluations during this specific backward pass is therefore 2. For a general $c^* > 2$, one needs to compute the backward pass for $c = 1, \dots, c^*-1$ without any early termination in order to obtain a valid reversible macro step.  Thus, when performing the backward pass for the values $c = 1, \dots, c^{*}-2$, the total number of gradient evaluations is $3(2^{c^* - 2}) - 2$. For $c = c^* - 1$, since $c^* > 2$, both the coarser and finer forward state of the backward pass for this value of $c$ have already been computed from the previous backward pass of $c = c^* - 2$ and the backward state from the final forward pass, respectively. Thus, only the backward state must be computed, which requires $2^{c^* - 2}$ gradient evaluations. The total number of gradient evaluations during the backward pass is then $3(2^{c^* - 2}) - 2 + 2^{c^* - 2} = 4(2^{c^* - 2}) - 2 = 2^2(2^{c^* - 2})-2 = 2^{c^*} - 2 $. 

Finally, for the RKNF criterion, the computation can be terminated early if an early micro step has already violated the criterion. This process is of course situation-dependent so that computing the exact number of gradient evaluations for both the forward and backward pass is not possible. However, it is possible to derive an upper bound of this cost. Given $c^*$, it is assumed that when computing the forward pass for $c = 0, \dots, c^*-1$, the criterion is violated at the very last micro step. Then, if $c^*=0$, the number of gradient evaluations for the forward pass is simply $1 + 1 = 2$. Because of Proposition \ref{prop:approx_reverse} (with Proposition \ref{prop_flow_criterion}), there is no need to perform the backward pass in this case. Given a general $c^* > 0$, for $c = 0, \dots, c^* - 1$, the maximum number of gradient evaluations performed while still violating the desired criterion is $2^c + 2^c = 2^{c+1}$, where the first term corresponds to the number of gradient evaluations due to the standard leapfrog step while the second term represents the number of gradient evaluations due to \eqref{eq:third_order_integrator_flow_nr2}. Thus, the maximum number of gradient evaluations becomes $\sum_{c=0}^{c^*} 2^{c+1} = 2^{c^* +2}-2$. Using the same assumption and following the identical strategy, the upper bound of the number of gradient evaluations for the backward pass is $2^{c^* +1}-2$. The final results of the computational cost analysis for the three criteria are summarized in Table \ref{tab:criteria_computational_cost}. 

\section{Strategies of tuning \texorpdfstring{$h$ and $\delta$}{h and delta}} \label{appendix:tuning_strategies}
As mentioned in Section \ref{sec:tuning_parameter_effect} and in some of the examples related to the multinomial HMC methods, the macro step size $h$ and the error threshold $\delta$ can be tuned to obtain a desired combination of $P(c = c_\text{min})$ and $P(\Delta < \Delta_\text{threshold})$. Now, further details regarding the tuning strategy similar to \citet{bou2025within} are discussed. 

\subsection{Tuning \texorpdfstring{$\delta$}{delta}}
Recall that for a general criterion, $\delta$ controls the number of micro steps needed to cover a full macro step. For a given iteration/orbit, this implies that it will also affect the largest deviation in the energy of the reachable states within an orbit, i.e., $\Delta$. Thus, similar to \citet{bou2025within}, this variation in the energy is here modeled as $\Delta \approx \lambda \delta$. 

Now, the goal is to choose $\delta$ such that $P(\Delta \leq \Delta_\text{threshold}) = \beta \in (0, 1)$. Based on the assumption above, during the warmup phase, the observed $\Delta$ is recorded for each iteration to compute the observed $\lambda$ corresponding to $\Delta_\text{obs} / \delta$. Then, $\delta$ is updated by $\Delta_\text{threshold} / q_\lambda(\beta)$, where $q_\lambda(\beta)$ corresponds to the empirical $\beta$-quantile of the recorded values of $\lambda$. 

\subsection{Tuning \texorpdfstring{$h$}{h}}
With a given $\delta$ and criterion, one can tune $h$ by relating it to $P(c = c_\text{min}) = \omega$. Because the leapfrog is a second-order integrator, it is plausible that for the coarsest step given $\delta$, $\delta = O(h^3)$ \citep{bou2025within}. Setting $\delta \approx Ch^3$ for some positive constant $C$ and similar to the above, one can record the observed value of $C$ in each iteration by $\kappa_\text{obs} / h^3$, where $\kappa_\text{obs}$ is the observed value of the chosen error criterion during a single macro step with $c = c_\text{min}$. Then, $h$ is updated according to $(\delta / q_C(\omega))^{1/3}$ with $q_C(\omega)$ representing the empirical $\omega$-quantile of the recorded values of $C$. 

\end{document}